\documentclass[12pt]{article}
\usepackage{amsmath,amssymb}
\usepackage[utf8]{inputenc}
\usepackage{etoolbox,fp}
\usepackage[dvipsnames,usenames]{xcolor}
\usepackage{tikz}
\usetikzlibrary{shapes}

\colorlet{acolor}{blue}
\colorlet{bcolor}{yellow}
\colorlet{ccolor}{green}
\colorlet{dcolor}{Dandelion}
\colorlet{aocolor}{Apricot!30}
\colorlet{aoocolor}{Apricot!40}
\colorlet{bocolor}{blue}
\colorlet{cocolor}{cyan}

\tikzset{
a/.style={->},
aa/.style={<-},
aaa/.style={->},
aarrow/.style={draw=blue, ultra thick,<-},
barrow/.style={draw=Red!70, ultra thick,<-},
  avertex/.style={circle,draw,inner sep=2pt,fill=acolor},
  smallavertex/.style={circle,draw,inner sep=1.5pt,fill=acolor},
  aavertex/.style={circle,draw,inner sep=2pt,fill=Apricot},
bvertex/.style={circle,draw,inner sep=2pt,fill=bcolor},
rvertex/.style={circle,draw,inner sep=2pt,fill=Red!70},
smallrvertex/.style={circle,draw,inner sep=1.5pt,fill=Red!70},
cvertex/.style={circle,draw,inner sep=2pt,fill=ccolor},
dvertex/.style={regular polygon,regular polygon sides=6,draw,inner sep=2pt,fill=dcolor},
  uvertex/.style={circle,draw=Black,inner sep=2pt,fill=white},
  schvertex/.style={regular polygon,regular polygon sides=6,draw,inner sep=2pt,fill=Red},
  vertex/.style={circle,draw,inner sep=2pt,fill=Gray!30},
  smallvertex/.style={circle,draw,inner sep=1.5pt,fill=Gray!30},
  aovertex/.style={circle,draw,inner sep=2pt,fill=aocolor},
  aoovertex/.style={circle,draw,inner sep=2pt,fill=aoocolor},
bovertex/.style={circle,draw,inner sep=2pt,fill=bocolor},
covertex/.style={circle,draw,inner sep=2pt,fill=cocolor},
dovertex/.style={regular polygon,regular polygon sides=6,draw,inner sep=2pt,fill=dcolor},
yvertex/.style={regular polygon,regular polygon sides=6,draw,inner sep=2pt,fill=ccolor}
}

\newtheorem{theorem}{Theorem}[section]

\newtheorem{lemma}[theorem]{Lemma}

\newtheorem{remark}[theorem]{Remark}
\newcommand{\bull}{\mbox{$\;\;\;$\vrule height .9ex width .8ex depth -.1ex}}
\newenvironment{proof}{\par\smallbreak\noindent{\bf Proof.~}}
{\unskip\nobreak\hfill \bull \par\medbreak}
\newenvironment{proofof}[1]{\par\smallbreak\noindent{\bf Proof of~#1.~}}
{\unskip\nobreak\hfill \bull \par\medbreak}
\newcounter{claim}
\renewcommand{\theclaim}{\Alph{claim}}
\newenvironment{claim}{\refstepcounter{claim}%
\par\medskip\par\noindent{\it Claim~\theclaim.~}~\rm}%
{\par\smallskip\par}
\newenvironment{subproof}{\par\noindent{\sl Proof of Claim~\theclaim.~}}%
{$\,\triangleleft$\par\medskip\par}
\newcommand{\noproofclaim}{$\,\triangleleft$}
\newcommand{\Case}[2]{\smallskip\par{\it Case #1:\/ #2.}}
\newcommand{\Subcase}[2]{\smallskip\par{\it Subcase #1:\/ #2.}}

\newcommand{\hide}[1]{}
\newcommand{\refeq}[1]{(\ref{eq:#1})}
\newcommand{\setdef}[2]{\left\{ \hspace{0.5mm} #1 : \hspace{0.5mm} #2 \right\}}
\newcommand{\function}[2]{:#1 \rightarrow #2}
\newcommand{\integers}{\mathbb{Z}}
\newcommand{\feq}{\stackrel{\mbox{\tiny def}}{=}}
\newcommand{\und}{\wedge}
\newcommand{\Or}{\vee}
\newcommand{\E}{\exists}
\newcommand{\A}{\forall}
\newcommand{\logic}{\mathcal{L}}
\newcommand{\fo}[1]{\mathrm{FO}^{#1}}
\newcommand{\conn}{\mathit{conn}}
\newcommand{\disc}{\mathit{disc}}
\newcommand{\thin}{\mathit{thin}}
\newcommand{\thick}{\mathit{thick}}
\newcommand{\norma}{\mathit{norma}}
\newcommand{\empt}{\mathit{empty}}
\newcommand{\compl}{\mathit{compl}}
\newcommand{\ac}[1]{\mbox{\rm AC$^{#1}$}}
\newcommand{\tc}[1]{\mbox{\rm TC$^{#1}$}}
\newcommand{\nc}[1]{\mbox{\rm NC$^{#1}$}}
\newcommand{\baru}{{\bar u}}
\newcommand{\barv}{{\bar v}}
\newcommand{\barw}{{\bar w}}
\newcommand{\tilu}{{\tilde u}}
\newcommand{\tilv}{{\tilde v}}
\newcommand{\barC}{{\bar C}}
\newcommand{\tilC}{{\tilde C}}

\DeclareMathOperator{\rank}{rk}
\newcommand{\dist}{\mathit{dist}}
\newcommand{\calB}{\mathcal{B}}
\newcommand{\hatT}{{\hat T}}
\newcommand{\Tau}{T}

\title{Bounds for the quantifier depth\\ in finite-variable logics: Alternation hierarchy}
\author{Christoph Berkholz\thanks{RWTH Aachen University, Institut f\"ur Informatik, D-52056 Aachen, Germany.}{}, 
Andreas Krebs\thanks{Wilhelm-Schickard-Institut, Universit\"at T\"ubingen, Sand 13, 72076 T\"ubingen, Germany.}{},
and Oleg Verbitsky\thanks{%
Humboldt-Universit\"at zu Berlin,
Institut f\"ur Informatik,
Unter den Linden 6,
D-10099 Berlin.
Supported by DFG grant VE 652/1--1.
On leave from the Institute for Applied Problems of Mechanics and Mathematics,
Lviv, Ukraine.}}

\date{}

\begin{document} 

\maketitle

\begin{abstract}
Given two structures $G$ and $H$ distinguishable in $\fo k$ (first-order logic with $k$ variables),
let $A^k(G,H)$ denote the minimum alternation depth of a $\fo k$ formula
distinguishing $G$ from $H$. Let $A^k(n)$ be the maximum value
of $A^k(G,H)$ over $n$-element structures. We prove the strictness
of the quantifier alternation hierarchy of $\fo 2$ in a strong quantitative
form, namely $A^2(n) > n/8-2$, which is tight up to a constant factor.
For each $k\ge2$, it holds that $A^k(n)>\log_{k+1}n-2$ even over colored trees, 
which is also tight up to a constant factor if $k\ge3$.
For $k\ge 3$ the last lower bound holds also over uncolored trees,
while the alternation hierarchy of $\fo 2$ collapses even over all uncolored graphs.

We also show examples of colored graphs $G$ and $H$ 
on $n$ vertices that can be distinguished in $\fo 2$
much more succinctly if the alternation
number is increased just by one: while in $\Sigma_{i}$ it is possible to distinguish $G$ from $H$ 
with bounded quantifier depth, in $\Pi_{i}$ this requires quantifier depth $\Omega(n^2)$.
The quadratic lower bound is best possible here because, if $G$ and $H$ can be distinguished
in $\fo k$ with $i$ quantifier alternations, this can be done with quantifier depth $n^{2k-2}$.
\end{abstract}

\section{Introduction}

Given structures $G$ and $H$ over vocabulary $\sigma$ and a first-order formula $\Phi$
over the same vocabulary, we say that $\Phi$ \emph{distinguishes} $G$ from $H$
if $\Phi$ is true on $G$ but false on $H$. By
\emph{alternation depth} of $\Phi$ we mean
the maximum length of a sequence of nested alternating quantifiers in $\Phi$.
Obviously, this parameter is bounded from above by the \emph{quantifier depth} of $\Phi$.
We will examine the maximum alternation depth and quantifier depth needed to distinguish two structures for restrictions of first-order logic and particular classes of structures.

For a fragment $\logic$ of first-order logic, by
$A_{\logic}(G,H)$ we denote the minimum alternation depth
of a formula $\Phi\in\logic$ distinguishing $G$ from $H$. 
Similarly, we let $D_{\logic}(G,H)$ denote the minimum quantifier depth of such $\Phi$.
Obviously, $A_{\logic}(G,H)\le D_{\logic}(G,H)$.
We define the \emph{alternation function} $A_{\logic}(n)$ to be equal to
the maximum value of $A_{\logic}(G,H)$ taken over all pairs of 
$n$-element structures $G$ and $H$ distinguishable in~$\logic$.

Our interest in this function is motivated
by the observation that if the quantifier alternation hierarchy of $\logic$
collapses, then $A_\logic(n)=O(1)$. 
More specifically, $A_\logic(n)\le a$ if the alternation hierarchy collapses 
to its $a$-th level $\Sigma_a\cup\Pi_a$.
Thus, showing that 
\begin{equation}
  \label{eq:Atoinf}
\lim_{n\to\infty}A_\logic(n)=\infty  
\end{equation}
is a way of proving that the hierarchy is strict. 

Note that Condition \refeq{Atoinf} is, in general, formally stronger than a hierarchy result.
For example, 
while the alternation hierarchy of first-order logic $\fo{}$
is strict over colored directed trees by Chandra and  Harel \cite{ChandraH82},
we have $A_{\fo{}}(n)=1$ for any class of structures over a fixed vocabulary.

\begin{figure}[b!]
\begin{center}
\renewcommand{\arraystretch}{1.2}
\begin{tabular}{l@{\extracolsep{5mm}}l@{\extracolsep{5mm}}l@{\extracolsep{5mm}}l}
\hline
\hline
Class of structures & Logic                      & Bounds for $A_\logic(n)$  & \\
\hline 
uncolored trees     & $\logic=\fo2$              & $=1$                 & Theorem \ref{thm:uncol-collapse} \\
\cline{2-4}
                    & $\logic=\fo{k}$, $k\geq 3$ & $>\log_{k+1} n-2$      & Theorem \ref{thm:hierarchy} \\
                    &                           & $<(k+3) \log_2 n$      & Theorem \ref{thm:log-upper-trees} \\
\hline
colored trees       & $\logic=\fo{k}$, $k\geq 2$ & $>\log_{k+1} n-2$      & Theorems \ref{thm:hierarchy2-color} and \ref{thm:hierarchy} \\
\cline{2-4}
                    & $\logic=\fo{k}$, $k\geq 3$ & $<(k+3) \log_2 n$      & Theorem \ref{thm:log-upper-trees} \\
\hline
uncolored graphs    & $\logic=\fo2$              & $=2$                & Theorem \ref{thm:uncol-collapse} \\
\cline{2-4}
                    & $\logic=\fo{k}$, $k\geq 3$ & $>\log_{k+1} n-2$      & Theorem \ref{thm:hierarchy} \\
                    &                            & $\le n^{k-1}+1$           & cf.~\cite{PikhurkoV11} \\
\hline
colored graphs      & $\logic=\fo2$              & $> n/8-2$             & Theorem \ref{thm:altcount} \\
                    &                            & $\le n+1$                 & cf.~\cite{ImmermanL90} \\
\cline{2-4}
                    & $\logic=\fo{k}$, $k\geq 3$ & $>\log_{k+1} n-2$      & Theorem \ref{thm:hierarchy} \\
                    &                            & $\le n^{k-1}+1$           & cf.~\cite{PikhurkoV11} \\
\hline 
\hline 
\end{tabular}
\end{center}
\vspace*{-5mm}
\caption{Results about $A_\logic(n)$}\label{fig:afunc}
\end{figure}

An example of this nature also exists when we restrict our logic to two variables:
While the alternation hierarchy of $\fo2[<]$
is strict over words in an infinite alphabet by Immerman and Weis \cite{WeisI09}, 
we have $A_{\fo2}(n)=1$ for words in any alphabet.

Moreover, the rate of growth of $A_{\logic}(n)$ can be naturally regarded as a quality 
of the strictness of the alternation hierarchy.
Note that any pair of structures $G$ and $H$ with $A_{\logic}(G,H)=a$
can serve as a certificate that the first $a$ levels of the alternation hierarchy
of $\logic$ are distinct.
Indeed, if $G$ is distinguished from $H$ by a formula $\Phi\in\logic$
of the minimum alternation depth $a$, then the set of structures $L=\setdef{S}{S\models\Phi}$
is not definable in $\logic$ with less than $a$ quantifier alternations.
Thus, the larger the value of $A_{\logic}(n)$ is, the more levels of the alternation
hierarchy can be separated by a certificate of size~$n$.

Results that we now know about the function $A_{\logic}(n)$ are displayed in Figure \ref{fig:afunc}.
The upper bound $A_{\fo k}(n)\le n^{k-1}+1$ holds true even for the quantifier depth.
It follows from the relationship of the distinguishability in $\fo k$ to the
$(k-1)$-dimensional color refinement (Weisfeiler-Lehman) procedure
discovered in \cite{ImmermanL90,CaiFI92} and the standard color stabilization argument; see \cite{PikhurkoV11}.
The logarithmic upper bound for trees (Theorem \ref{thm:log-upper-trees}) holds true also for the quantifier depth.

Additionally, in Section \ref{sec:lower} we show that the $\Sigma_i$ fragment of $\fo2$ is not only strictly more expressive than
the $\Sigma_{i-1}$ fragment but also more succinct
in the following sense:
There are colored graphs $G$ and $H$
on $n$ vertices such that they can be distinguished in $\Sigma_{i-1}\cap\fo2$
and, moreover, this is possible with
bounded quantifier depth in $\Sigma_{i}\cap\fo2$ while in $\Pi_{i}\cap\fo2$ this requires quantifier depth $\Omega(n^2)$.
The quadratic lower bound is best possible here because, if $G$ and $H$ can be distinguished
in $\fo k$ with $i$ quantifier alternations, this can be done with quantifier depth $n^{2k-2}$ 
(Section~\ref{sec:upper}).

\section{Preliminaries}\label{sec:prelim}

\subsection{Notation}

We consider first-order formulas only in the negation normal form
(i.e., any negation stands in front of a relation symbol
and otherwise only monotone Boolean connectives are used).
For each $i\ge1$, let $\Sigma_i$ (resp.\ $\Pi_i$) denote the set of
(not necessary prenex) formulas where any sequence of nested
quantifiers has at most $i-1$ quantifier alternations
and begins with $\exists$ (resp.\ $\forall$).
In particular, \emph{existential logic} $\Sigma_1$ consists of formulas
without universal quantification.
Up to logical equivalence, $\Sigma_i\cup\Pi_i\subset\Sigma_{i+1}\cap\Pi_{i+1}$.
By the \emph{quantifier alternation hierarchy} we mean
the interlacing chains $\Sigma_1\subset\Sigma_2\subset\ldots$
and $\Pi_1\subset\Pi_2\subset\ldots$. We are interested in the
corresponding fragments of a finite-variable logic.

As a short notation we use $D^k_\logic(G,H)=D_{\logic\cap\fo k}(G,H)$
and $A^k_\logic(G,H)=A_{\logic\cap\fo k}(G,\allowbreak H)$.
The subscript $\fo{}$ can be dropped; for example, $D^k(G,H)=D^k_{\fo{}}(G,H)$ and $A^k(n)=A^k_{\fo{}}(n)$.
Sometimes we will write $D^k_{\E}(G,H)$ in place of $D^k_{\Sigma_1}(G,H)$
or $D^k_{\A}(G,H)$ in place of $D^k_{\Pi_1}(G,H)$.

The universe of a structure $G$ will be denoted by $V(G)$,
and the number of elements in $V(G)$ will be denoted by $v(G)$.
Since binary structures can be regarded as vertex- and edge-colored directed graphs,
the elements of $V(G)$ will also be called vertices.
A vertex in a simple undirected graph is \emph{universal} if it is adjacent to all other vertices.

\subsection{The Ehrenfeucht-Fra\"iss\'e game}

The \emph{$k$-pebble Ehrenfeucht-Fra\"iss\'e game on structures $G$ and $H$},
is played by two players, Spoiler and Duplicator,
to whom we will refer as he and she respectively.
The players have equal sets of $k$ pairwise different pebbles.
A {\em round\/} consists of a move of Spoiler followed by a move of
Duplicator. Spoiler takes a pebble and puts it on a vertex in $G$ or in $H$.
Then Duplicator has to put her copy of this pebble on a vertex
of the other graph. Duplicator's objective is to keep the following
condition true after each round: the pebbling should determine a partial
isomorphism between $G$ and $H$.
The variant of the game where Spoiler starts playing  
in $G$ and is allowed to jump from one graph to the other less than $i$ times
during the game will be referred to as the \emph{$\Sigma_i$ game}.
In the \emph{$\Pi_i$ game} Spoiler starts in~$H$.

For each positive integer $r$, the $r$-round Ehrenfeucht-Fra\"iss\'e game (as well as its $\Sigma_i$ and 
$\Pi_i$ variants)
is a two-person game of perfect information with a finite number of positions.
Therefore, either Spoiler or Duplicator has a \emph{winning strategy}
in this game, that is, a strategy winning against every strategy of the opponent.

\begin{lemma}[e.g., \cite{WeisI09}]
  \label{lem:game}
$D^k_{\Sigma_i}(G,H)\le r$ if and only if
Spoiler has a winning strategy in the $r$-round $k$-pebble $\Sigma_i$ game on $G$ and~$H$.
\end{lemma}

\subsection{The lifting construction}\label{ss:lifting}

Note that separation of the ground floor of the alternation hierarchy for $\fo2$ costs nothing.
We can take graphs $G$ and $H$ with three isolated vertices each, color one vertex of $G$ 
in red, and color the other vertices of $G$ and all vertices of $H$ in blue.
Obviously, $D^2_\E(G,H)=1$ while $D^2_\A(G,H)=\infty$.
It turns out that any separation example can be lifted to higher floors in a rather
general way. Similarly, given a sequence of examples of $n$-vertex graphs $G$ and $H$
with $D^2_\E(G,H)=\Omega(n^2)$ and $D^2_\A(G,H)=O(1)$, in Section \ref{sec:lower} we will be able to lift it
to any number of quantifier alternations.

%We now aim at showing that the quadratic lower bound of Theorem \ref{thm:lower-E} for existential
%two-variable logic can be lifted to any floor of the quantifier alternation hierarchy of $\fo2$.
The lifting gadget provided by Lemma \ref{lem:lifting} below is a reminiscence
of the classical construction designed by Chandra and Harel to prove
the strictness of the first-order alternation hierarchy. The
Chandra-Harel construction 
is applicable to other logics (see, e.g., \cite[Section 8.6.3]{EbbinghausF})
and can be used as
a general scheme for obtaining hierarchy results.
This approach was also used by Oleg Pikhurko (personal communication, 2007)
to construct, for each $i$, a sequence of pairs of trees $G_n$ and $H_n$
such that $D_{\Sigma_i}(G_n,H_n)=O(1)$ while $D_{\Pi_i}(G_n,H_n)\to\infty$
as $n\to\infty$.

\begin{figure}
\centering
\begin{tikzpicture}
  \begin{scope}[scale=2]
    \node at (-1.1 cm,-.3 cm) {$G_1$};
    \matrix[nodes={draw},column sep=0.5cm,row sep=1cm] {
    \node[circle, inner sep=1mm] {$G_0$}; &
    \node[circle, inner sep=1mm] {$H_0$}; &
    \node[circle, inner sep=1mm] {$H_0$};\\
&\node[vertex] (n) {};&\\
    };
\coordinate (a) at (-.84,.13);
\coordinate (b) at (.84,.13);
\coordinate (x) at (-.4,.05);
\coordinate (y) at (-.15,.05);
\coordinate (z) at (.4,.05);
\coordinate (w) at (.15,.05);
\draw (n)--(a) (n)--(b) (n)--(x) (n)--(z) (n)--(y) (n)--(w);
  \end{scope}
  \begin{scope}[scale=2,xshift=30mm]
    \node at (-1.1 cm,-.3 cm) {$H_1$};
    \matrix[nodes={draw},column sep=0.5cm,row sep=1cm] {
    \node[circle, inner sep=1mm] {$H_0$}; &
    \node[circle, inner sep=1mm] {$H_0$}; &
    \node[circle, inner sep=1mm] {$H_0$};\\
&\node[vertex] (n) {};&\\
    };
\coordinate (a) at (-.84,.13);
\coordinate (b) at (.84,.13);
\coordinate (x) at (-.4,.05);
\coordinate (y) at (-.15,.05);
\coordinate (z) at (.4,.05);
\coordinate (w) at (.15,.05);
\draw (n)--(a) (n)--(b) (n)--(x) (n)--(z) (n)--(y) (n)--(w);
  \end{scope}

\newcommand{\hi}{.9}
\newcommand{\wi}{.6}
\newcommand{\tri}[2]{
\node[vertex] (g#2) at (0,0) {};
\node at (0,.6) {$#1$};
\draw (g#2) -- (-\wi,\hi) -- (\wi,\hi) -- (g#2);
}
\newcommand{\triii}[3]{
    \matrix[ampersand replacement=\&,column sep=3mm,row sep=7mm] {
   \tri{#1}{1}  \&  \tri{#2}{2}  \& \tri{#3}{3} \\
\&\node[vertex] (g) {};\&\\
    };
\draw (g) -- (g1) (g) -- (g2) (g) -- (g3) ;
}
  \begin{scope}[yshift=-30mm]
    \node at (-2.5 cm,0 cm) {$G_{i+1}$};
\triii{G_i}{G_i}{H_i}
  \end{scope}
  \begin{scope}[xshift=60mm,yshift=-30mm]
    \node at (-2.5 cm,0 cm) {$H_{i+1}$};
\triii{G_i}{G_i}{G_i}
  \end{scope}
\end{tikzpicture}
\caption{The lifting construction.}
\label{fig:lifting}
\end{figure}
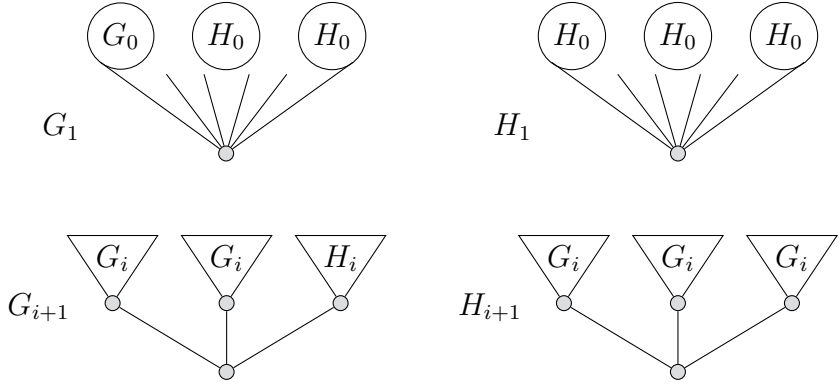

Given colored graphs $G_0$ and $H_0$, we recursively construct graphs $G_i$ and $H_i$
as shown in Fig.~\ref{fig:lifting}. $H_1$ consists of three disjoint copies of $H_0$
and an extra universal vertex, that will be referred to as the root vertex of $H_1$.
The root vertex is colored in a new color absent in $G_0$ and $H_0$, say, in gray.
The graph $G_1$ is constructed similarly but, instead of three $H_0$-branches,
it has two $H_0$-branches and one $G_0$-branch.
Suppose that $i\ge1$ and the rooted graphs $G_{i}$ and $H_{i}$ are already constructed.
The graph $H_{i+1}$ consists of three disjoint copies of $G_{i}$
and the gray root vertex adjacent to the root of each $G_{i}$-part.
The graph $G_{i+1}$ is constructed similarly but, instead of three $G_{i}$-branches,
it has two $G_{i}$-branches and one $H_{i}$-branch.

We will say that Spoiler plays \emph{continuously} if, after each of his moves,
the two pebbled vertices are adjacent.

\begin{lemma}\label{lem:lifting}
  Assume that Spoiler has a continuous strategy allowing him to win
the 2-pebble $\Sigma_1$ game on $G_0$ and $H_0$ in $r$ moves. Then, for each $i\ge1$,
\begin{enumerate}
\item 
$D^2_{\Sigma_i}(G_i,H_i) < r+i$;
\item 
$D^2_{\Sigma_i}(G_i,H_i)\ge D^2_{\Pi_{i+1}}(G_i,H_i)\ge D^2_\E(G_0,H_0)$;
\item 
$D^2_{\Pi_i}(G_i,H_i)=\infty$;
\item 
If, moreover, Spoiler has a continuous strategy allowing him to win
the 2-pebble $\Sigma_2$ game on $G_0$ and $H_0$ in $s$ moves, then
$D^2_{\Sigma_{i+1}}(G_i,H_i) < s+i$.
\end{enumerate}
\end{lemma}

\begin{proof}
{\bf 1.}  
In the base case of $i=1$ we have to show that Spoiler is able to win the $\Sigma_1$ game
on $G_1$ and $H_1$ in $r$ moves. He forces the $\Sigma_1$ game
on $G_0$ and $H_0$ by playing continuously inside the $G_0$-part of $G_1$ and wins by assumption.

Now, we recursively describe a strategy for Spoiler in the $\Sigma_{i+1}$ game
on $G_{i+1}$ and $H_{i+1}$ and inductively prove that it is winning.
For each $i$, the strategy will be continuous, and the vertex pebbled in the
first round will be adjacent to the root. Note that this is true in the base case.

In the first round Spoiler pebbles the root of the $H_i$-branch of $G_{i+1}$.
Duplicator is forced to pebble the root of one of the $G_i$-branches of $H_{i+1}$.
Indeed, if she pebbles a gray vertex at the different distance from the root of $H_{i+1}$,
then Spoiler pebbles a shortest possible path upwards in $G_{i+1}$ or $H_{i+1}$
and wins once he reaches a non-gray vertex. In the second round Spoiler jumps to
this $G_i$-branch and, starting from this point, forces the $\Sigma_{i}$ game
on $G_{i}$ and $H_{i}$ by playing recursively and, hence, continuously.
The only possibility for Duplicator to avoid the recursive play and not to
lose immediately is to pebble a gray vertex below. In this case Spoiler wins
in altogether $i+1$ moves by pebbling a path upwards in the graph where he stays,
as already explained. If the game goes recursively, then by the induction assumption
Spoiler needs less than $1+r+i$ moves to win.

{\bf 2.}  
In the base case of $i=1$ we have to design a strategy for Duplicator
in the $\Pi_2$ game on $G_1$ and $H_1$. First of all, Duplicator pebbles
the gray vertex always when Spoiler does so. Furthermore, whenever
Spoiler pebbles a vertex in an $H_0$-branch of $G_1$ or $H_1$, Duplicator pebbles the same vertex
in an $H_0$-branch of the other graph. It is important that, if the pebbles are
in two different $H_0$-branches of $G_1$ or $H_1$, Duplicator has a possibility
to pebble different $H_0$-branches in the other graph. It remains to describe
Duplicator's strategy in the case that Spoiler moves in the $G_0$-branch of $G_1$.
Note that once Spoiler does so, he cannot change the graph any more.
In this case, Duplicator chooses a free $H_0$-branch in $H_1$ and follows her
optimal strategy in the $\Sigma_1$ game on $G_0$ and $H_0$. Since the gray vertex
is universal in both graphs and the $G_0$- and $H_0$-branches are isolated from
each other, Spoiler wins only when he wins the $\Sigma_1$ game on $G_0$ and $H_0$,
which is possible in $D^2_\E(G_0,H_0)$ moves at the earliest.

In the $\Pi_{i+2}$ game on $G_{i+1}$ and $H_{i+1}$ Duplicator plays similarly.
She always respects the root vertex, the $G_i$-branches, and takes care that
the pebbled vertices are either in the same or in distinct $G_i$-branches
in both graphs. Once Spoiler moves in the $H_i$-branch of $G_{i+1}$,
Duplicator invokes her optimal strategy in the $\Sigma_{i+1}$ game on $H_i$ and $G_i$,
what is the same as the $\Pi_{i+1}$ game on $G_i$ and $H_i$.
There is no other way for Spoiler to win than to win this subgame.
By the induction assumption, this takes at least $D^2_\E(G_0,H_0)$ moves.

{\bf 3.}  
By induction on $i$, we show that Duplicator has a strategy allowing her
to resist arbitrarily long in the $\Pi_i$ game on $G_i$ and $H_i$.
An important feature of the strategy is that Duplicator will always respect
the distance of a pebbled gray vertex from the root. 
In the base case of $i=1$, such a strategy exists because in $G_1$ there are
two copies of $H_0$, where Duplicator can mirror Spoiler's moves.
In the $\Pi_{i+1}$ game on $G_{i+1}$ and $H_{i+1}$, Duplicator makes use of
the existence of two copies of $G_i$ in both graphs. Whenever Spoiler
pebbles the root vertex or moves in a $G_i$-part in any of $G_{i+1}$ and $H_{i+1}$,
Duplicator mirrors this move in the other graph. Whenever Spoiler
moves for the first time in the $H_i$-part of $G_{i+1}$, Duplicator responds
in a free $G_i$-part of $H_{i+1}$ according to her level-preserving strategy for 
the $\Pi_i$ game on $G_i$ and $H_i$, that exists by the induction assumption.
When Spoiler moves in the $H_i$-part also with the other pebble,
Duplicator continues playing in the same $G_i$-part of $H_{i+1}$
following the same strategy.

{\bf 4.}  
Spoiler has a recursive winning strategy for the $\Sigma_{i+1}$ game on $G_i$ and $H_i$
similarly to the proof of part~1.
\end{proof}

\section{Alternation function for $\fo k$ over trees}

\begin{theorem}\label{thm:hierarchy2-color}
$A^2(n)>\log_3n-2$ over colored trees.
\end{theorem}

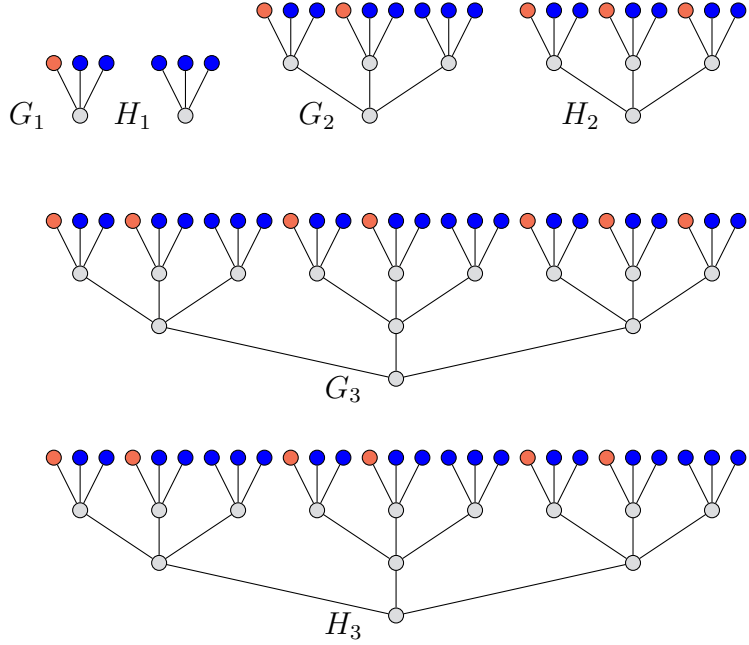
\begin{figure}
\centering
 \begin{tikzpicture}[scale=.7]

\newcommand{\hlevel}{1}

\newcommand{\hone}[2]{
\path (#1,#2) node[vertex] (r) {}
     +(-.5,\hlevel) node[avertex] (r1)  {} edge (r) 
       +(0,\hlevel) node[avertex] (r2)  {} edge (r) 
      +(.5,\hlevel) node[avertex] (r3)  {} edge (r);
}
\newcommand{\gone}[2]{
\path (#1,#2) node[vertex] (r) {}
     +(-.5,\hlevel) node[rvertex] (r1)  {} edge (r) 
       +(0,\hlevel) node[avertex] (r2)  {} edge (r)
      +(.5,\hlevel) node[avertex] (r3)  {} edge (r);
}
\newcommand{\gtwo}[2]{
\node[vertex] (rr) at (#1,#2) {};
\draw (rr) -- ++(-1.5,\hlevel) (rr) -- ++(0,\hlevel) (rr) -- ++(1.5,\hlevel);
\gone{#1}{#2+\hlevel}
\gone{#1-1.5}{#2+\hlevel}
\hone{#1+1.5}{#2+\hlevel}
}
\newcommand{\htwo}[2]{
\node[vertex] (rr) at (#1,#2) {};
\draw (rr) -- ++(-1.5,\hlevel) (rr) -- ++(0,\hlevel) (rr) -- ++(1.5,\hlevel);
\gone{#1}{#2+\hlevel}
\gone{#1-1.5}{#2+\hlevel}
\gone{#1+1.5}{#2+\hlevel}
}

\begin{scope}
\node at (-1cm,0cm) {$G_1$};
  \gone{0}{0}
\end{scope}
\begin{scope}[xshift=2 cm]
\node at (-1cm,0cm) {$H_1$};
  \hone{0}{0}
\end{scope}

\begin{scope}[xshift=5.5cm]
\node at (-1cm,0cm) {$G_2$};
  \gtwo{0}{0}
\end{scope}

\begin{scope}[xshift=10.5cm]
\node at (-1cm,0cm) {$H_2$};
  \htwo{0}{0}
\end{scope}

\newcommand{\gthree}[2]{
\node[vertex] (rrr) at (#1,#2) {};
\draw (rrr) -- ++(-4.5,\hlevel) (rrr) -- ++(0,\hlevel) (rrr) -- ++(4.5,\hlevel);
\gtwo{#1}{#2+\hlevel}
\gtwo{#1-4.5}{#2+\hlevel}
\htwo{#1+4.5}{#2+\hlevel}
}
\newcommand{\hthree}[2]{
\node[vertex] (rrr) at (#1,#2) {};
\draw (rrr) -- ++(-4.5,\hlevel) (rrr) -- ++(0,\hlevel) (rrr) -- ++(4.5,\hlevel);
\gtwo{#1}{#2+\hlevel}
\gtwo{#1-4.5}{#2+\hlevel}
\gtwo{#1+4.5}{#2+\hlevel}
}
\begin{scope}[xshift=6cm,yshift=-5cm]
\node at (-1cm,-.2cm) {$G_3$};
  \gthree{0}{0}
\end{scope}
\begin{scope}[xshift=6 cm,yshift=-9.5cm]
\node at (-1cm,-.2cm) {$H_3$};
  \hthree{0}{0}
\end{scope}
 \end{tikzpicture}
\caption{Proof of Theorem \protect\ref{thm:hierarchy2-color}}
\label{fig:trees}
\end{figure}

\begin{proof}
Applying the lifting construction described in Section \ref{sec:prelim} to 
the pair of single-vertex graphs
$$G_0=
\begin{tikzpicture}
  \node[rvertex] at (0,0) {};
\end{tikzpicture}
\text{\quad and\quad}
H_0=
\begin{tikzpicture}
  \node[avertex] at (0,0) {};
\end{tikzpicture},
$$
we obtain the sequence of pairs of colored trees $G_i$ and $H_i$ with $v(G_i)=v(H_i)$
as shown in Fig.~\ref{fig:trees}.
For $i\ge1$, we have
$$
D^2_{\Sigma_i}(G_i,H_i)\le i
$$
by part 1 of Lemma \ref{lem:lifting} and
$$
D^2_{\Pi_i}(G_i,H_i)=\infty
$$
by part 3 of this lemma. It follows that $A^2(n_i)\ge i$ for $n_i=v(G_i)$.
Note that $n_{i}=3n_{i-1}+1$, where $n_0=1$. Therefore
$n_{i}=3^i+\frac{3^i-1}2$, which implies that $A^2(n_i)>\log_3n_i-1$. 

Consider now an arbitrary $n$ and suppose that $n_i\le n<n_{i+1}$, i.e.,
$n_i\le n\le3n_i$. We can increase the number of vertices in $G_i$ and $H_i$
to $n$ by attaching $n-n_i$ new gray leaves at the root. Since this
does not change the parameters $D^2_{\Sigma_i}(G_i,H_i)$ and $D^2_{\Pi_i}(G_i,H_i)$,
we get $A^2(n)\ge A^2(n_i)>\log_3n-2$.
\end{proof}

Theorem \ref{thm:hierarchy2-color} generalizes to any $k$-variable logic
and, if $k>2$, then no vertex coloring is needed any more.

\begin{theorem}\label{thm:hierarchy}
If $k\ge3$, then
$A^k(n)>\log_{k+1}n-2$ over uncolored trees.
\end{theorem}

\begin{proof}
Notice that the lifting construction of Lemma \ref{lem:lifting} generalizes to
$k\ge3$ variables by adding $k-2$ extra copies of $H_0$ in $G_1$ and $H_1$
and $k-2$ extra copies of $G_i$ in $G_{i+1}$ and $H_{i+1}$. Similarly to Theorem \ref{thm:hierarchy2-color},
this immediately gives us colored trees $G_{i+1}$ and $H_{i+1}$ such that
$D^2_{\Sigma_i}(G_i,H_i)\le i$ and $D^k_{\Pi_i}(G_i,H_i)=\infty$ for all $i\ge1$; see Fig.~\ref{fig:GH12-still-uncolored}.

\begin{figure}
\centering
 \begin{tikzpicture}[scale=.8]

\newcommand{\hlevel}{1}

\newcommand{\hone}[2]{
\path (#1,#2) node[vertex] (r) {}
     +(-.75,\hlevel) node[avertex] (r1)  {} edge (r) 
       +(-.25,\hlevel) node[avertex] (r2)  {} edge (r) 
      +(.25,\hlevel) node[avertex] (r3)  {} edge (r)
      +(.75,\hlevel) node[avertex] (r4)  {} edge (r);
}
\newcommand{\gone}[2]{
\path (#1,#2) node[vertex] (r) {}
     +(-.75,\hlevel) node[rvertex] (r1)  {} edge (r) 
       +(-.25,\hlevel) node[avertex] (r2)  {} edge (r) 
      +(.25,\hlevel) node[avertex] (r3)  {} edge (r)
      +(.75,\hlevel) node[avertex] (r4)  {} edge (r);
}
\newcommand{\gtwo}[2]{
\node[vertex] (rr) at (#1,#2) {};
\draw (rr) -- ++(-3,\hlevel) (rr) -- ++(-1,\hlevel) (rr) -- ++(1,\hlevel) (rr) -- ++(3,\hlevel);
\gone{#1-3}{#2+\hlevel}
\gone{#1-1}{#2+\hlevel}
\gone{#1+1}{#2+\hlevel}
\hone{#1+3}{#2+\hlevel}
}
\newcommand{\htwo}[2]{
\node[vertex] (rr) at (#1,#2) {};
\draw (rr) -- ++(-3,\hlevel) (rr) -- ++(-1,\hlevel) (rr) -- ++(1,\hlevel) (rr) -- ++(3,\hlevel);
\gone{#1-3}{#2+\hlevel}
\gone{#1-1}{#2+\hlevel}
\gone{#1+1}{#2+\hlevel}
\gone{#1+3}{#2+\hlevel}
}

\begin{scope}[xshift=-3 cm]
\node at (-1cm,0cm) {$G_1$};
  \gone{0}{0}
\end{scope}
\begin{scope}%[xshift=.5 cm]
\node at (-1cm,0cm) {$H_1$};
  \hone{0}{0}
\end{scope}

\begin{scope}[yshift=-3.5cm]
\node at (-1cm,0cm) {$G_2$};
  \gtwo{0}{0}
\end{scope}

\begin{scope}[yshift=-6.5cm]
\node at (-1cm,0cm) {$H_2$};
  \htwo{0}{0}
\end{scope}

 \end{tikzpicture}
\caption{Proof of Theorem \protect\ref{thm:hierarchy}.
The trees for 3-variable logic are still colored.}
\label{fig:GH12-still-uncolored}
\end{figure}
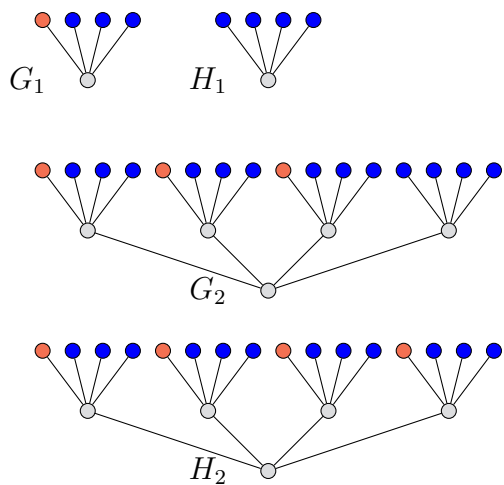

In order to remove colors from $G_i$ and $H_i$, we construct these graphs recursively in the same way but
now, instead of red and blue one-vertex graphs, we start with
$$G_0=
\begin{tikzpicture}[every node/.style=uvertex,scale=.5]
\path (-.5,1) node (a) {}
      (0,0) node (b)  {} edge (a) 
       +(.5,1) node (c)  {} edge (b);
\end{tikzpicture}
\text{\quad and\quad}
H_0=
\begin{tikzpicture}[every node/.style=uvertex,scale=.5]
\path (0,1.5) node (a) {}
      (0,.75) node (b)  {} edge (a) 
       (0,0) node (c)  {} edge (b);
\end{tikzpicture},
$$
see Fig.~\ref{fig:GH1uncolored}.
Note that in the course of construction $G_0$ and $H_0$ will be handled as
rooted trees (otherwise they are isomorphic).

\begin{figure}
\centering
 \begin{tikzpicture}[every node/.style=uvertex,scale=.7]
\newcommand{\hlevel}{1}
\begin{scope}[xscale=-1]
\node[draw=none,fill=none] at (1cm,0cm) {$G_1$};
\path (0,0) node (r) {}
      +(-1.5,\hlevel) node (a)  {} edge (r) 
       +(-.5,\hlevel) node (b)  {} edge (r) 
      +(.5,\hlevel) node (c)  {} edge (r)
      +(1.5,\hlevel) node (d)  {} edge (r)
     +(-1.5,2*\hlevel) node (a1)  {} edge (a) 
     +(-.5,2*\hlevel) node (b1)  {} edge (b) 
     +(.5,2*\hlevel) node (c1)  {} edge (c)
     +(-1.5,3*\hlevel) node ()  {} edge (a1) 
     +(-.5,3*\hlevel) node ()  {} edge (b1) 
     +(.5,3*\hlevel) node ()  {} edge (c1)
     +(1,2*\hlevel) node ()  {} edge (d)
     +(2,2*\hlevel) node ()  {} edge (d);
\end{scope}
\begin{scope}[xshift=5.5cm]
\node[draw=none,fill=none] at (-1cm,0cm) {$H_1$};
\path (0,0) node (r) {}
     +(-1.5,\hlevel) node (a)  {} edge (r) 
       +(-.5,\hlevel) node (b)  {} edge (r) 
      +(.5,\hlevel) node (c)  {} edge (r)
      +(1.5,\hlevel) node (d)  {} edge (r)
     +(-1.5,2*\hlevel) node (a1)  {} edge (a) 
     +(-.5,2*\hlevel) node (b1)  {} edge (b) 
     +(.5,2*\hlevel) node (c1)  {} edge (c)
     +(1.5,2*\hlevel) node (d1)  {} edge (d)
    +(-1.5,3*\hlevel) node ()  {} edge (a1) 
     +(-.5,3*\hlevel) node ()  {} edge (b1) 
     +(.5,3*\hlevel) node ()  {} edge (c1)
     +(1.5,3*\hlevel) node ()  {} edge (d1);
\end{scope}
 \end{tikzpicture}
\caption{Proof of Theorem \protect\ref{thm:hierarchy}.
The uncolored versions of $G_1$ and $H_1$ for 3-variable logic.}
\label{fig:GH1uncolored}
\end{figure}
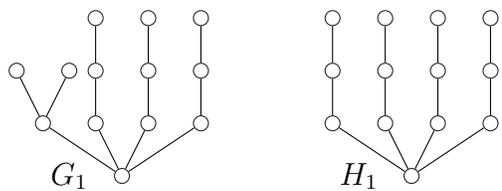

We now claim that for the uncolored trees $G_i$ and $H_i$ it holds
\begin{enumerate}
\item 
$D^3_{\Sigma_i}(G_i,H_i)\le i+5$,
\item 
$D^k_{\Pi_i}(G_i,H_i)=\infty$.
\end{enumerate}
The latter claim is true exactly by the same reasons as in the colored case:
since the number of Spoiler's jumps is bounded, Duplicator is always able to
ensure playing on isomorphic branches.
To prove the former claim, we will show that Spoiler can win similarly to the colored case
playing with 3 pebbles.

Note that in the uncolored version of $G_i$ and $H_i$, all formerly gray vertices have degree $k+1$,
red vertices have degree 3, and blue vertices have degree 2.
A typical ending of the game on the colored trees was that Spoiler pebbles a red vertex
while Duplicator is forced to pebble a blue one. Now this corresponds to pebbling
a vertex $u$ of degree 3 by Spoiler and a vertex $v$ of degree 2 by Duplicator.
Having 4 pebbles, Spoiler would win by 
pebbling the three neighbors of $u$. Having only 3 pebbles, Spoiler
 first pebbles two neighbors $u_1$ and $u_2$ of $u$
(in fact, one neighbor is already pebbled immediately before $u$). 
Duplicator must respond with the two
neighbors $v_1$ and $v_2$ of $v$. In the next round Spoiler moves the pebble from
$u$ to its third neighbor $u_3$. Duplicator must remove the pebble from $v$
and place it on some vertex $v_3$ non-adjacent to both $v_1$ and $v_2$.
Note that, while the distance between any two vertices of $u_1$, $u_2$, and $u_3$
equals 2, there is a pair of indices $s$ and $t$ such that $v_s$ and $v_t$
are at the distance more than 2. Spoiler now wins by moving the pebble
from $u_q$ to $u$, where $\{q\}=\{1,2,3\}\setminus\{s,t\}$.

It remains to note that with 3 pebbles Spoiler is able to force climbing upwards
in the trees and, hence, he can follow essentially the same winning strategy as in the colored
case. Duplicator can deviate from this scenario only in the first round.
Recall that in this round Spoiler pebbles a vertex $u$ at the distance 1 from the root level,
having degree at least 3. Suppose that
Duplicator responds with pebbling a vertex $v$ at the distance more than 1 from the
root level. If $i=1$, then $v$ is of degree at most 2, and Spoiler wins as explained above.
If $i\ge2$, then $v$ can have degree 3 or $k+1$. In this case
Spoiler forces climbing up and wins by pebbling a leaf above a formerly blue 
vertex because by this point Duplicator has already reached the highest possible level.
Suppose now that in the first round Duplicator
pebbles the root vertex. Then Spoiler puts a second pebble on the root
of his graph, Duplicator is forced to pebble a vertex one level higher,
and Spoiler again wins by forcing climbing up from the root to the highest
leaf level.

Thus, we have shown that $A^k(n_i)\ge i$ for $n_i=v(G_i)$.
Since $n_{i}=(k+1)n_{i-1}+1$ and $n_0=3$, we have
$n_{i}=3(k+1)^i+\frac{(k+1)^i-1}k$, which implies that $A^2(n_i)>\log_{k+1}n_i-1$. 
Like to the proof of Theorem \ref{thm:hierarchy2-color}, this bound extends to
all $n$ at the cost of decreasing it by~$1$.
\end{proof}

\begin{remark}\rm
 The proof of Theorem \ref{thm:hierarchy} implies that a limited number of quantifier alternations
cannot be compensated by an increased number of variables:
for every $k\ge3$ and $i\ge1$ there is a class of uncolored graphs definable in $\fo3$
but not in $\Sigma_i\cap\fo k$. If we allow vertex colors, there is such a class
definable even in $\fo2$.
\end{remark}

 Theorems \ref{thm:hierarchy2-color} and \ref{thm:hierarchy} are optimal in the sense that
they cannot be extended to $\fo2$ over uncolored trees.
The reason is that
 the quantifier alternation hierarchy of $\fo2$ over uncolored graphs
collapses to the second level; see Section~\ref{sec:uncol-collapse}.

We now show that the bound of Theorem \ref{thm:hierarchy} is tight up
to a constant factor. The following theorem implies that, if $k\ge3$,
then $A^k(n)<(k+3)\log_2n$ over colored trees. The proof easily extends
to the class of all binary structures whose Gaifman graph is a tree.

\begin{theorem}\label{thm:log-upper-trees}
Let $k\ge3$. If $D^k(T,T')<\infty$ for colored trees $T$ and $T'$, then
  \begin{equation}
    \label{eq:TTlog}
D^k(T,T')<(k+3)\log_2n   
  \end{equation}
 where $n$ denotes the number of vertices in~$T$.
\end{theorem}

\begin{proof}
Let $T-v$ denote the result of removal of a vertex $v$ from the tree $T$.
The component of $T-v$ containing a neighbor $u$ of $v$ will be
denoted by $T_{vu}$ and considered a rooted tree with the root at $u$.
A similar notation will apply also to $T'$.
The rooted trees $T_{vu}$ will be called \emph{branches of $T$ at the vertex $v$}.
Let $\tau(v)$ denote the maximum number of pairwise isomorphic branches at $v$.
We define the \emph{branching index} of $T$ by $\tau(T)=\max_v\tau(v)$.
In order to prove the theorem, we will show that
the bound \refeq{TTlog} is true for any non-isomorphic colored trees
with branching index at most $k$ and that
 $D^k(T,T')=D^k(T\bmod k,T'\bmod k)$ for $T\bmod k$ and $T'\bmod k$
being ``truncated'' versions of $T$ and $T'$ whose branching index
is bounded by $k$. We first handle the latter task.

The following fact easily follows from the trivial observation that $k$
pebbles can be placed on at most $k$ isomorphic branches.

\begin{claim}\label{cl:truncate}
  Let $T$ be a colored tree. Suppose that $T$ has more than $k$ isomorphic
branches at a vertex $v$. Remove all but $k$ of them from $T$ and denote
the resulting tree by $\hatT$. Then $D^k(T,G)=D^k(\hatT,G)$ for any colored graph~$G$.
\noproofclaim
\end{claim}

The truncated tree $T\bmod k$ is obtained from $T$ by a series of truncations
as in Claim \ref{cl:truncate}. The truncations steps should be done from the top
to the bottom in order to exclude appearance of new isomorphic branches
in the course of the procedure. In order to define the ``top and bottom'' formally, recall that
the {\em eccentricity\/} of a vertex $v$ in a graph $G$ is defined
by $e(v)=\max_{u}\dist(v,u)$, where $\dist(v,u)$ denotes the distance between the two vertices. The {\em diameter\/} and the
{\em radius\/} of $G$ are defined by $d(G)=\max_{v}e(v)$
and $r(G)=\min_{v}e(v)$ respectively.
%A path in a graph is {\em diametral\/} if its length is equal to
%the diameter of the graph. 
A vertex $v$ is {\em central\/} if $e(v)=r(G)$.
For trees it is well known (e.g., \cite[Chapter 4.2]{Ore}) that
if $d(T)$ is even, then $T$ has a unique central
vertex $c$.
% and all diametral paths go through $c$.
If $d(T)$ is odd, then $T$ has exactly two central vertices $c_1$ and
$c_2$, that are adjacent.
% and all diametral paths go through the edge $\{c_1,c_2\}$.
Let us regard the central vertices as lying on the bottom level
and the tree $T$ as growing upwards. The height of a vertex is then
its distance to the nearest central vertex. Starting from the highest
level and going downwards, for each vertex $v$ we cut off extra branches
at $v$ if their number exceeds $k$. Note that this operation can
increase the number of isomorphic branches from vertices in lower levels
but cannot do this for vertices in higher levels. Therefore,
the resulting tree $T\bmod k$ has branching index at most~$k$.

Applying repeatedly Claim \ref{cl:truncate}, we arrive at the equality
$D^k(T,T')=D^k(T\bmod k,T'\bmod k)$. Note that $T\bmod k\not\cong T'\bmod k$
because it is assumed that $D^k(T,T')<\infty$. Thus, we have reduced proving
the bound \refeq{TTlog} to the case that $T$ and $T'$ are non-isomorphic
and both have branching index at most $k$. Therefore, below we make this assumption.

We have to show that Spoiler is able to win the $k$-pebble game on such
$T$ and $T'$ in less than $(k+3)\log_2n$ moves.
Below we will actively exploit the following fact ensured by a standard halving strategy
for Spoiler.

\begin{claim}\label{cl:distance}
Suppose that in the 3-pebble Ehrenfeucht-Fra\"iss\'e game on graphs $G$ and $H$ some two
vertices $x,y\in V(G)$ at distance $n$ are pebbled so that their
counterparts $x',y'\in V(H)$ are at a strictly larger distance.
Then Spoiler can win in at most $\lceil\log n\rceil$ extra moves.
\noproofclaim
\end{claim}

Every tree $T$ has a single-vertex \emph{separator}, that is, a vertex $v$ such that
no branch of $T$ at $v$ has more than $n/2$ vertices;
see, e.g., \cite[Chapter 4.2]{Ore}.
The idea of Spoiler's strategy is to pebble such a vertex and
to force further play on some non-isomorphic branches of $T$ and $T'$,
where the same strategy can be applied recursively.
This scenario was realized in \cite[Theorem 5.2]{PikhurkoV11}
for first-order logic with counting quantifiers.
Without counting, we have to use some additional tricks
that are based on boundedness of the branching index.
Below, by $N(v)$ we will denote the neighborhood of a vertex~$v$.

Thus, in the first round Spoiler pebbles a separator $v$ in $T$
and Duplicator responds with a vertex $v'$ somewhere in $T'$.
Since $T\not\cong T'$, there is an isomorphism type $\calB$
of a branch of $T$ at $v$ that appears with different multiplicity
among the branches of $T'$ at $v'$. Spoiler can use this fact
to force pebbling vertices $u\in N(v)$ and $u'\in N(v')$
so that the rooted trees $T_{vu}$ and $T'_{v'u'}$ are non-isomorphic
(the pebbles on $v$ and $v'$ can be reused but, finally, $v$ and $v'$
have to remain pebbled as well).
This is easy to do if the multiplicity of $\calB$ in one of the trees
is at most $k-2$. If this multiplicity is $k-1$ in one tree
and $k$ in the other, then Spoiler can do it still with $k$ pebbles
like as in the proof of Theorem \ref{thm:hierarchy}.
This phase of the game can take $k+2$ rounds.

The next goal of Spoiler is to force pebbling adjacent vertices $v_1$ and $u_1$
in $T_{vu}$ and adjacent vertices $v'_1$ and $u'_1$ in $T'_{v'u'}$ so that
$T_{v_1u_1}\not\cong T'_{v'_1u'_1}$ and 
\begin{equation}
  \label{eq:halving}
 v(T_{v_1u_1})\le v(T_{vu})/2\text{ or }
v(T'_{v'_1u'_1})\le v(T_{vu})/2.
\end{equation}
Once this is done, the same will be repeated recursively
(with the roles of $T$ and $T'$ swapped if only the second inequality in \refeq{halving} is true).

To make the transition from $T_{vu}$ to $T_{v_1u_1}$, Spoiler first pebbles a separator $w$ of $T_{vu}$. 
Note that Duplicator
is forced to respond with a vertex $w'$ in $T'_{v'u'}$.
Otherwise we would have $\dist(w,u)=\dist(w,v)-1$ while $\dist(w',u')=\dist(w',v')+1$.
Therefore, some distances among the three pebbled vertices
would be different in $T$ and in $T'$ and Spoiler could win
in less than $\log v(T_{vu})+1$ moves by Claim~\ref{cl:distance}.

Let $T_{w\setminus u}$ denote the rooted tree obtained by removing from $T$ the branch
at $w$ containing $u$ and rooting the resulting tree at $w$. Note that $V(T_{w\setminus u})\subset V(T_{vu})$.
We consider a few cases.

\Case{1}{$T_{w\setminus u}\not\cong T'_{w'\setminus u'}$}
In the trees $T_{w\setminus u}$ and $T'_{w'\setminus u'}$ we will consider branches at their roots $w$ and~$w'$.

\Subcase{1-a}{$T_{w\setminus u}$ contains a branch of isomorphism type $\calB$ that 
has different multiplicity in $T'_{w'\setminus u'}$}
As above, Spoiler can use $k$ pebbles and $k+1$ moves to force pebbling vertices
$x\in N(w)$ and $x'\in N(w')$ such that $T_{wx}\not\cong T'_{w'x'}$ and
\begin{equation}
  \label{eq:calB}
 T_{wx}\in\calB\text{ or }T'_{w'x'}\in\calB.
\end{equation}
The pebbles occupying $v,v'$ and $u,u'$ can be released.
The pebbles on $w$ and $w'$ can also be reused but, finally, $w$ and $w'$
have to remain pebbled.
 The branches $T_{wx}$ and $T'_{w'x'}$ will now
serve as $T_{v_1u_1}$ and $T'_{v'_1u'_1}$. 
Condition \refeq{halving} follows from \refeq{calB} because $w$ is a separator of~$T_{vu}$.
% If the second inclusion in \refeq{calB} is true, then it is not excluded that $T_{wx}$
% contains $u$ and, hence, $v(T_{wx})>v(T_{vu})$. In this case the roles of $T$ and $T'$ are
% swapped (Spoiler jumps from $T$ to~$T'$).

\Subcase{1-b}{$T_{w\setminus u}$ does not contain any branch as in Subcase 1-a}
In this subcase there is a vertex $x'\in N(w')$ such that $T'_{w'x'}$
is a branch of $T'_{w'\setminus u'}$ and the isomorphism type of $T'_{w'x'}$
does not appear in $T_{w\setminus u}$. Spoiler moves the pebble from $v'$ to $x'$.
Suppose that Duplicator responds with $x\in N(w)$.
If $x$ lies on the path between $u$ and $w$ (while $x'$ does not lie on the path between $u'$ and $w'$),
then equality of distances among the pebbled vertices cannot be preserved, and Spoiler
wins by Claim~\ref{cl:distance}. If $x$ does not lie between $u$ and $w$, then
$T_{wx}$ is a branch of $T_{vu}$ at the vertex $w$. The first equality in Condition \refeq{halving}
is then true because $w$ is a separator of $T_{vu}$. In this case, 
$T_{wx}$ and $T'_{w'x'}$ can serve as $T_{v_1u_1}$ and $T'_{v'_1u'_1}$.

\Case{2}{$T_{w\setminus u}\cong T'_{w'\setminus u'}$}
We assume that $\dist(u,w)=\dist(u',w')$ because otherwise Spoiler wins by Claim \ref{cl:distance}.
For a vertex $y$ on the path between $u$ and $w$, let $T_{y\setminus u,w}$ denote 
the rooted tree obtained by removing from $T$ the branches
at $y$ containing $u$ and $w$ and rooting the resulting tree at $y$.
The rooted tree $T_{u\setminus v,w}$ is defined similarly.
Note that $T_{u\setminus v,w}$ and each $T_{y\setminus u,w}$
are parts of a branch of $T_{vu}$ at the vertex $w$ and, therefore, have
at most $v(T_{vu})/2$ vertices. Given $y$ between $u$ and $w$,
by $y'$ we will denote the vertex lying between $u'$ and $w'$
at the same distance to these vertices as $y$ to $u$ and $w$.
Since $T_{vu}\not\cong T'_{v'u'}$, we must have
\begin{equation}
  \label{eq:nonisoy}
 T_{y\setminus u,w}\not\cong T'_{y'\setminus u',w'} 
\end{equation}
for some $y$ or
\begin{equation}
  \label{eq:nonisou}
 T_{u\setminus v,w}\not\cong T'_{u'\setminus v',w'}.
\end{equation}
Assume that Condition \refeq{nonisoy} is true for some $y$ and fix this vertex.

\Subcase{2-a}{$T_{y\setminus u,w}$ contains a branch of isomorphism type $\calB$ that 
has different multiplicity in $T'_{y'\setminus u',w'}$}
Spoiler moves the pebble from $v$ to $y$. Duplicator is forced to move
the pebble from $v'$ to $y'$. The pebbles occupying $u,u'$ and $w,w'$ can now be released.
Spoiler proceeds similarly to Subcase 1-a and forces pebbling 
vertices $z\in N(y)$ and $z'\in N(y')$ such that $T_{yz}\not\cong T'_{y'z'}$ and
one of these trees has isomorphism type $\calB$ and, hence, is as small as desired.

\Subcase{2-b}{$T_{y\setminus u,w}$ does not contain any branch as in Subcase 2-a}
In this subcase there is a vertex $z'\in N(y')$ such that $T'_{y'z'}$
is a branch of $T'_{y'\setminus u',w'}$ whose isomorphism type
does not appear in $T_{y\setminus u,w}$. Similarly to Subcase 1-b,
Spoiler aims to pebble $y'$ and $z'$ while forcing Duplicator to respond with 
$y$ and $z\in N(y)$ such that $T_{yz}$ is a part of $T_{y\setminus u,w}$.
This will ensure that $T_{yz}\not\cong T'_{y'z'}$ and that $T_{yz}$
is small enough. Now Spoiler's task is more complicated because
he has to prevent Duplicator from pebbling $z$ on the path between $u$ and $w$.
Since this requires keeping the pebbles on $u,u'$ and $w,w'$, Spoiler cannot pebble
both $y'$ and $z'$ if there are only $k=3$ pebbles. In this case he first pebbles
the vertex $z'$ by the pebble released from $v$. Let $z$ be Duplicator's response. 
If $z$ is in $N(y)$ and does not lie between $u$ and $w$, Spoiler succeeds by
moving the pebble from $u'$ to $y'$. Duplicator is forced to 
move the pebble from $u$ to $y$ because $w'$ remains pebbled and, therefore,
the position of $y$ is determined by the distances to $z$ and $w$.
If $z$ is not in $N(y)$ or lies between $u$ and $w$, then Spoiler wins
because $\dist(z,u)\ne\dist(z',u')$ or $\dist(z,w)\ne\dist(z',w')$

An analysis of the case \refeq{nonisou}
is quite similar. The role of the triple $(u,y,w)$ is now played
by the triple $(v,u,w)$.

Note that the transition from $T_{vu}$ to $T_{v_1u_1}$ takes
at most $k+3$ rounds.
Also, 2 rounds suffice to win the game
once the current subtree $T_{vu}$ has at most 2 vertices. 
The number of transitions from the initial branch of order at most
$n/2$ to one with at most $2$ vertices is bounded by 
$\log_2n-1$ because $v(T_{vu})$ becomes twice smaller each time. 
It follows that Spoiler wins the game on $T$ and $T'$ in less than
$k+2+(\log_2n-1)(k+3)+2\le(k+3)\log_2n+1$ moves.
The additive term of 1 can be dropped because if pebbling the initial branch
takes no less than $k+2$ moves, then the size of this branch will actually
not exceed~$n/k$.
\end{proof}

\section{Alternation function for $\fo 2$ over colored graphs}

Theorem \ref{thm:hierarchy2-color} gives us a logarithmic
 lower bound on the alternation function $A^2(n)$, which is true even for trees. 
Over all colored graphs, we now prove a linear lower bound.
Along with the general upper bound $A^2(n)\le n+1$,
it shows that $A^2(n)$ has a linear growth.

\begin{theorem}\label{thm:altcount}
$A^2(n) > n/8-2$.
\end{theorem}

\begin{proof}
For each integer $m\ge2$, we will construct colored graphs $G$ and $H$, both with $n=8m-4$ vertices,
that can be distinguished in $\fo2$ with $m-2$, but no less than that, alternations.
The graph $G=2G_m$ is the union of two disjoint copies of the same graph $G_m$
and, similarly, $H=2H_m$ where $G_m$ and $H_m$ are defined as follows.
Each of $G_m$ and $H_m$ is obtained by merging
two building blocks $A_m$ and $B_m$ shown in Fig.~\ref{fig:altcount}.
The colored graph $A_m$ is a ``ladder'' with $m$ horizontal rungs, each having
2 vertices. The vertices on the bottom rung are colored in green,
the vertices on the top rung are colored one in red and the other
in blue, the remaining $2m-4$ vertices are white (uncolored). The graph $B_m$
is obtained from $A_m$ by recoloring red in apricot and blue in cyan.
$A_m$ and $B_m$ are glued together at the green vertices.
There are two ways to do this, and the resulting graphs $G_m$ and $H_m$
are non-isomorphic. Let $\alpha^+$ (resp.\ $\alpha^-$)
denote the partial isomorphism from $G_m$ to $H_m$ identifying the $A_m$-parts
(resp.\ the $B_m$-parts) of these graphs.

We will design a strategy allowing Spoiler to win the $(m-2)$-alternation
(i.e., $\Sigma_{m-1}$ or $\Pi_{m-1}$)
2-pebble Ehrenfeucht-Fra\"iss\'e game on $G$ and $H$ and a strategy
allowing Duplicator to win the $(m-3)$-alternation game.
Before playing on $G$ and $H$, we analyse the 2-pebble game on $G_m$ and $H_m$.
Spoiler can win this game as follows. In the first round
he pebbles the left green vertex in $G_m$; see Fig.~\ref{fig:altcount}. 
Not to lose immediately, Duplicator responds either with the left or with the right
green vertex in $H_m$. The corresponding
partial isomorphism can be extended to $\alpha^+$ in the former case and to $\alpha^-$
in the latter case (but not to both $\alpha^+$ and $\alpha^-$). These two cases are similar, and we
consider the latter of them, where there is no extension to $\alpha^+$ and hence 
Spoiler has a chance to win playing in the $A_m$-parts of $G_m$ and~$H_m$.

\newcommand{\drawN}{
     \node[\lv] (a\i) at (0*\r,\r*\i) {};
     \node[\rv] (b\i) at (1*\r,\r*\i) {};
}

\newcommand{\drawE}{ 
     \FPset{\j}{\i}
     \FPadd{\i}{\i}{1}
     \FPround{\i}{\i}{0}
     \drawN
     \draw  (a\i) -- (a\j);
     \draw  (b\i) -- (b\j);
}

\newcommand{\drawL}{ 
     \drawE
     \draw  (a\i) -- (b\j);
}

\newcommand{\drawR}{ 
     \drawE
     \draw  (b\i) -- (a\j);
}

\begin{figure}
\centering
 \begin{tikzpicture}
\draw[<->,shorten <=5pt, shorten >=5pt,very thin] (1,3) arc (180:360:.75);
\draw[<->,shorten <=5pt, shorten >=5pt,very thin] (0,3) arc (180:360:1.75);
\draw[<->,shorten <=9pt, shorten >=9pt,dashed] (.1,3) arc (180:360:1.25);
\draw[<->,shorten <=9pt, shorten >=9pt,dashed] (.9,3) arc (180:360:1.25);
 \newcommand{\lv}{uvertex}
 \newcommand{\rv}{uvertex}
\begin{scope}[yshift=30mm]
   \node at (-.6cm,0 cm) {$A_4$};
   \FPset{\r}{1}
   \FPset{\i}{0}
   \renewcommand{\lv}{yvertex}  % y = green  
   \renewcommand{\rv}{yvertex}  % y = green 
   \drawN
   \renewcommand{\lv}{uvertex}
   \renewcommand{\rv}{uvertex}
   \drawR
   \drawL
   \renewcommand{\lv}{rvertex}   % r = red
   \renewcommand{\rv}{bovertex} % bo = blue
   \drawR
 \end{scope}
\begin{scope}[yshift=30mm,xshift=25mm]
   \node at (-.6cm,0 cm) {$B_4$};
   \FPset{\r}{1}
   \FPset{\i}{0}
   \renewcommand{\lv}{yvertex}  % y = green  
   \renewcommand{\rv}{yvertex}  % y = green 
   \drawN
   \renewcommand{\lv}{uvertex}
   \renewcommand{\rv}{uvertex}
   \drawR
   \drawL
   \renewcommand{\lv}{aoovertex}  % aoo = apricot
   \renewcommand{\rv}{covertex}  % co = cyan
   \drawR
 \end{scope}
 \begin{scope}[xshift=60mm]
  \node at (-.6cm,0 cm) {$G_4$};
   \FPset{\r}{1}
   \FPset{\i}{0}
   \renewcommand{\lv}{aoovertex} % aoo = apricot
   \renewcommand{\rv}{covertex} % co = cyan
   \drawN
   \renewcommand{\lv}{uvertex}
  \renewcommand{\rv}{uvertex}
   \drawL
   \drawR
   \renewcommand{\lv}{yvertex}  % y = green  
   \renewcommand{\rv}{yvertex}  % y = green 
   \drawL
   \renewcommand{\lv}{uvertex}
   \renewcommand{\rv}{uvertex}
   \drawR
   \drawL
   \renewcommand{\lv}{rvertex}  % r = red
   \renewcommand{\rv}{bovertex} % bo = blue
   \drawR
 \end{scope}
 \begin{scope}[xshift=90mm]
  \node at (-.6cm,0 cm) {$H_4$};
   \FPset{\r}{1}
   \FPset{\i}{0}
   \renewcommand{\rv}{aoovertex} % aoo = apricot
   \renewcommand{\lv}{covertex}  % co = cyan
   \drawN
   \renewcommand{\lv}{uvertex}
   \renewcommand{\rv}{uvertex}
   \drawR
   \drawL
   \renewcommand{\lv}{yvertex}  % y = green  
   \renewcommand{\rv}{yvertex}  % y = green 
   \drawR
   \renewcommand{\lv}{uvertex}
   \renewcommand{\rv}{uvertex}
   \drawR
   \drawL
   \renewcommand{\lv}{rvertex}  % r = red
   \renewcommand{\rv}{bovertex} % bo = blue
   \drawR
 \end{scope}

 \end{tikzpicture}
\caption{Proof of Theorem \protect\ref{thm:altcount}: Construction of $G=2G_4$ and $H=2H_4$.}\label{fig:altcount}
\end{figure}
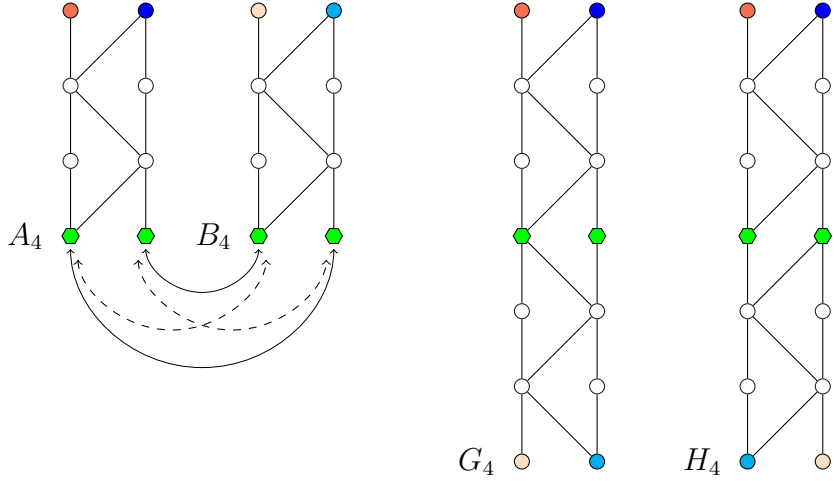

In the second round Spoiler pebbles the upright neighbor of the left green
vertex in $G_m$.
His goal in subsequent rounds is to force pebbling, one by one, edges along
the upright paths to the red vertex in $G_m$ and to the blue vertex in $H_m$.
If Duplicator makes a step down, Spoiler wins by reaching the top rung sooner
than Duplicator. If Duplicator moves all the time upward,
starting from the third round of the game she has a possibility to slant.
Spoiler prevents this by changing the graph. Note that in one of the graphs
there is only one way upstairs, and Spoiler always leaves this graph for Duplicator.
In this way Spoiler wins by making $m$ moves and alternating between
the graphs $m-2$ times.

The strategy we just described is inoptimal with respect to the alternation number.
In fact, Spoiler can win the game on $G_m$ and $H_m$ with no alternation at all
by pebbling in the first round the right green vertex in $G_m$. If Duplicator
responds with the left green vertex in $H_m$, Spoiler puts the second pebble
on the non-adjacent vertex in the next upper rung. Duplicator is forced to play
in a different rung of $H_m$ because otherwise she would violate the non-adjacency relation.
If in the first round Duplicator responds with the right green vertex, Spoiler plays
similarly, but in the lower rung of $G_m$. In any case, the second pebble is closer
to the red or to the apricot vertex in $G_m$ than in $H_m$, which makes Spoiler's win easy.

Nevertheless, the former, $(m-2)$-alternation strategy has an advantage:
Spoiler ensures that the two pebbled vertices are always adjacent.
By this reason, the same strategy can be used by Spoiler to win also the game
on $G=2G_m$ and $H=2H_m$. Once Duplicator steps aside to another copy of $G_m$ or $H_m$,
she immediately loses.

The partial isomorphism $\alpha^+$ from $G_m$ to $H_m$ determines two partial isomorphisms
$\alpha^+_0$ and $\alpha^+_1$ from $G=2G_m$ to $H=2H_m$ identifying the two $A_m$-parts
of $G$ with the two $A_m$-parts of $H$. Similarly, $\alpha^-$ gives rise to two
partial isomorphisms $\alpha^-_0$ and~$\alpha^-_1$.

We now show that the number of alternations $m-2$ is optimal for the game on $G$ and $H$.
Fix an integer $a$ such that Spoiler has a winning strategy in the $a$-alternation
2-pebble game on $G$ and $H$. For this game,  
let us fix an arbitrary winning strategy for Spoiler and
a strategy for Duplicator satisfying the following conditions.
\begin{itemize}
\item
Duplicator always respects vertex rungs (this is clearly possible
because every rung in $G$ and $H$ has four vertices and there are only two pebbles).
\item
Additionally, Duplicator respects adjacency (this is possible and complies with
the preceding rule because there are edges only between adjoining rungs
and every vertex sends at least one edge to each adjoining rung).
\item 
Duplicator respects also non-adjacency. Moreover, whenever Spoiler violates
adjacency of the vertices pebbled in one graph, Duplicator responds so that
the vertices pebbled in the other graph are not only non-adjacent but even
lie in different $G_m$- or $H_m$-components.
\item
If Spoiler pebbles a vertex above the green rung and the three preceding rules
still do not determine Duplicator's response uniquely, then she responds according
to $\alpha^+_0$ or $\alpha^+_1$; in a similar situation below the green rung, she plays according
to $\alpha^-_0$ or~$\alpha^-_1$.
\end{itemize}
Note that these rules uniquely determine Duplicator's moves on non-green vertices
provided one pebble is already on the board. In particular, the choice of $\alpha^+_0$ or $\alpha^+_1$
in the last rule depends on the component where this pebble is placed.

Let $u_i\in V(G)$ and $v_i\in V(H)$ denote the vertices pebbled in
the $i$-th round of the game. 
We now highlight a crucial property of Duplicator's strategy.
Suppose that $u_i,v_i$ and $u_{i+1},v_{i+1}$ are in the $A_m$-parts of $G$ and $H$
and that $u_{i+1}$ and $v_{i+1}$ are non-green. 
Then the following conditions are met.
\begin{itemize}
\item 
If $u_i$ and $u_{i+1}$ (as well as $v_i$ and $v_{i+1}$) are non-adjacent, then $\alpha^+_s(u_{i+1})=v_{i+1}$
for $s=0$ or $s=1$.
\item 
If $u_i$ and $u_{i+1}$ (as well as $v_i$ and $v_{i+1}$) are adjacent
and $\alpha^+_s(u_i)=v_i$ for $s=0$ or $s=1$, then $\alpha^+_s(u_{i+1})=v_{i+1}$
for the same~$s$.
\end{itemize}
The similar property holds if the pebbles are in the $B_m$ parts.

Suppose that Spoiler wins in the $r$-th round.
Note that Duplicator's strategy allows Spoiler to win only when
$u_r$ and $v_r$ are on the top or on the bottom rungs and have different
colors (in the absence of these colors,
the described strategy would be winning for Duplicator).
Since the two cases are similar, assume that Spoiler wins on the top.

Let $p$ be the smallest index such that all vertices in the sequence
$u_p,v_p,\ldots,u_r,v_r$ are above the green level.
By assumption, $\alpha^+_s(u_r)\ne v_r$ for both $s=0,1$.
The aforementioned property of Duplicator's strategy implies that, furthermore,
\begin{equation}\label{eq:ineq}
\alpha^+_0(u_i)\ne v_i\text{ and }\alpha^+_1(u_i)\ne v_i\text{ for all }i\ge p.
\end{equation}
Therefore, $u_{i+1}$ and $u_i$
as well as $v_{i+1}$ and $v_i$ are adjacent for all $i\ge p$
(for else Duplicator plays so that $\alpha^+_s(u_{i+1})=v_{i+1}$ for $s=0$ or $s=1$).
By the same reason, $p>1$ and $u_{p-1}$ and $u_p$ are also adjacent.
It follows that $u_{p-1}$ and $v_{p-1}$ are green and $\alpha^+_s(u_{p-1})\ne v_{p-1}$ for both $s=0,1$.

Another consequence of \refeq{ineq} is that both vertex sequences
$u_{p-1},u_p,\ldots,u_r$ and $v_{p-1},v_p,\ldots,v_r$ lie
on upright paths. This follows from the fact that either from $u_i$ 
or from $v_i$ there is only one edge emanating upstairs (also downstairs),
and it is upright.

It remains to notice that after each transition to the adjoining rung
(i.e., from $u_i,v_i$ to $u_{i+1},v_{i+1}$ for $i\ge p-1$)
Spoiler has to jump to the other graph because otherwise Duplicator
will choose the neighbor that ensures $\alpha^+_s(u_{i+2})=v_{i+2}$
for some value of $s=0,1$ (determined by the condition that $\alpha^+_s$ is a map from the component
of $G$ containing $u_{p-1}$ to the component of $H$ containing $v_{p-1}$).
This observation readily implies that the number of alternations $a$ cannot be smaller than~$m-2$.

We have shown that $A^2(n)\ge m-1$ if $n=8m-4$. Adding up to seven
isolated vertices to both $G$ and $H$, we get the same bound also
for $n=8m-3,\ldots,8m+3$. Therefore, $A^2(n)\ge(n-11)/8$ for all~$n$.
\end{proof}

%\section{Lower bounds for the distinctiveness index and succinctness results}\label{sec:lower}
\section{Lower bounds for $D^k_\logic(G,H)$ and succinctness results}\label{sec:lower}

\subsection{Existential two-variable logic}

%By Theorem \ref{thm:upper-bound},
In the next section we will see that, if binary structures $G$ and $H$ have $n$ elements each
and $G$ is distinguishable from $H$ in existential two-variable logic, then
$D^2_\E(G,H)\le n^2+1$. 
Here we show that this bound is tight up to a constant factor.
For the existential-positive fragment of $\fo2$, a quadratic lower bound can be
obtained from the benchmark instances for the arc consistency problem going back to
\cite{Dechter.1985,Samal.1987}; see \cite{BerkholzV13} where also an alternative approach
is suggested. We here elaborate on the construction presented in \cite{BerkholzV13}.
To implement this idea for existential two-variable logic, we need to undertake a more delicate analysis
as the existential-positive fragment is more restricted and simpler.

\begin{theorem}\label{thm:lower-E}
  There are infinitely many colored graphs $G$ and $H$, both on $n$ vertices,
such that $G$ is distinguishable from $H$ in existential two-variable logic and
$$\textstyle
D^2_\E(G,H)>\frac1{11}\,n^2.
$$
\end{theorem}

\begin{proof}
Our construction will depend on an integer parameter $m\ge2$.
We construct a pair of colored graphs $G_m$ and $H_m$ such that
$G_m$ is distinguishable from $H_m$ in the existential two-variable logic,
both $v(G_m)=O(m)$ and $v(H_m)=O(m)$, and $D^2_\E(G_m,H_m)=\Omega(m^2)$.
Though $v(G_m)<v(H_m)$, later we will be able to increase the number of vertices
in $G_m$ to $v(H_m)$.

The graphs have vertices of 4 colors, namely apricot, blue, cyan, and dandelion.
$G_m$ contains a cycle of length $3(2m-1)$ where apricot, blue, and cyan
alternate in this order; see Fig.~\ref{fig:main-example}. 
$H_m$ contains a similar cycle of length $3\cdot2\,m$.
Successive apricot, blue, and cyan vertices will be denoted by
$a_i$, $b_i$, and $c_i$ in $G_m$, where $0\le i<2m-1$, and by
$a'_i$, $b'_i$, and $c'_i$ in $H_m$, where $0\le i\le2m-1$.
Furthermore, the vertex $a_0$ is adjacent to a dandelion vertex $d_0$,
and every $a'_i$ except for $i=m$ is adjacent to a dandelion vertex $d'_i$.
This completes the description of the graphs.

\begin{figure}
\centering
 \begin{tikzpicture}[scale=.7]
  \FPset{\r}{2}
  \node at (-\r cm,\r cm + 1cm) {$G_3$};
\begin{scope}[xscale=-1]
  \foreach \i in {0,...,4}
   {
    \node[aovertex] (a\i) at (\i*72:\r cm) {};
    \FPsub{\j}{\i}{1}
    \FPround{\j}{\j}{0}
    \ifnumgreater{\i}{0} 
      {\draw (a\i) -- (c\j);}{}
     \node[bovertex] (b\i) at (\i*72+24:\r cm) {};
    \node[covertex] (c\i) at (\i*72+48:\r cm) {};
    \draw (a\i) -- (b\i) -- (c\i); 
   }
   \node[dovertex] (d) at (\r + 1,0) {};
   \draw (c4) -- (a0) -- (d);
   \node[right] (egal) at (a0) {$\,\,a_0$};
\end{scope}

  \FPset{\p}{2.2}
  \node[xshift=70mm] at (-\p cm - 1cm,\r cm + 1cm) {$H_3$};
\begin{scope}[xshift=70mm]
  \foreach \i in {0,...,5}
   {
    \node[aovertex] (a\i) at (\i*60:\p cm) {};
    \FPsub{\j}{\i}{1}
    \FPround{\j}{\j}{0}
    \ifnumgreater{\i}{0} 
      {\draw (a\i) -- (c\j);
       \node[dovertex] (d\i) at (\i*60:\p cm + 1cm) {};
       \draw (a\i) -- (d\i);
      }{}
     \node[bovertex] (b\i) at (\i*60+20:\p cm) {};
    \node[covertex] (c\i) at (\i*60+40:\p cm) {};
    \draw (a\i) -- (b\i) -- (c\i); 
   }
    \draw (c5) -- (a0);
   \node[right] (egal) at (a0) {$\,\,a'_3$};
\end{scope}
 \end{tikzpicture}
\caption{Proof of Theorem \protect\ref{thm:lower-E}. 
The only way for Spoiler to win is to force pebbling the pair of vertices
$(a_0,a'_3)$ in a round.
Duplicator can force Spoiler to make many passes around $G_m$ before this
goal is achieved.}
\label{fig:main-example}
\end{figure}
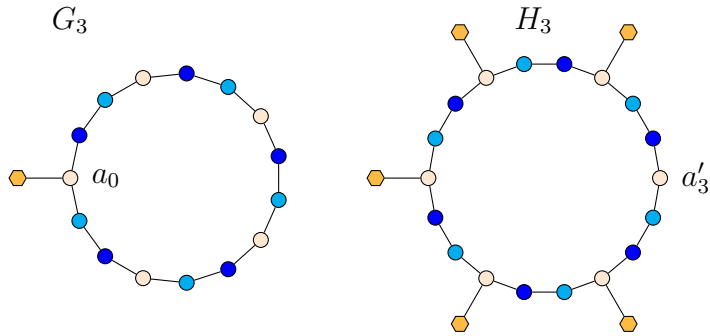

By Lemma \ref{lem:game}, we have to show that Spoiler is able
to win the 2-pebble $\Sigma_1$ game on $G_m$ and $H_m$ and that Duplicator is able to prevent losing
the game for $\Omega(m^2)$ rounds.

Note that, once the pair $(a_0,a'_m)$ is pebbled, Spoiler wins in the next move
by pebbling $d_0$. He is able to force pebbling $(a_0,a'_m)$ as follows.
In the first round he pebbles $a_0$. Suppose that Duplicator responds with
$a'_s$, where $0\le s<m$. In a series of subsequent moves,
Spoiler goes around the whole circle in $G_m$, visiting $c_{2m-2},b_{2m-2},a_{2m-2},c_{2m-2},\ldots$
and using the two pebbles alternately
(if $m<s<2m$, he does the same but in the other direction). 
As Spoiler comes back to $a_0$, Duplicator is forced to arrive at $a'_{s+1}$.
The next Spoiler's tour around the circle brings Duplicator to $a'_{s+2}$ and so forth.
Thus, the most successful moves for Duplicator in the first round is $a'_0$.
Then Spoiler needs to play $1+m\cdot3(2m-1)+1=6m^2-3m+2$ rounds in order to win.

Our next task is to design a strategy for Duplicator allowing her
to survive $\Omega(m^2)$ rounds, no matter how Spoiler plays.
We will show that Duplicator is able
to force Spoiler to pass around the cycle in $G_m$ many times.
A crucial observation is that $(a_0,a'_m)$ is the only pair whose
pebbling allows Spoiler to win in one extra move.

Let us regard the additive group $\integers_{2m}$ as a cycle graph
with $i$ and $j$ adjacent iff $i-j=\pm1$. Denote the distance
between vertices in this graph by $\Delta$.
The same letter will denote the following partial function
$\Delta\function{V(G_m)\times V(H_m)}{\integers}$.
For two vertices of the same color, say, for $a_i$ and $a'_j$,
we set $\Delta(a_i,a'_j)=\Delta(i,j)$.
Note that $\Delta(a_0,a'_m)=m$, which is the largest possible value.
Duplicator' strategy will be to keep the value of the $\Delta$-function
on the pebbled pair as small as possible.

Specifically, in the first round Duplicator responds to Spoiler's
move $x$ with pebbling a vertex $x'$ such that $\Delta(x,x')=0$
(that is, if $x=a_i,b_i,c_i,d_0$, then $x'=a'_i,b'_i,c'_i,d'_0$ respectively).
Suppose that a pair $(y,y')$ is pebbled in the preceding round
and Duplicator is still alive.
If Spoiler pebbles $x$ in the current round, Duplicator chooses
her response $x'$ by the following criteria.
Below, $\sim$ denotes the adjacency relation.
\begin{itemize}
\item 
$x'$ should have the same color as $x$ and, moreover,
$x'\sim y'$ iff $x\sim y$
(this is always possible unless $(y,y')=(a_0,a'_m)$ and $x=d_0$);
\item 
if there is still more than one choice, $x'$ should minimize
the parameter $\Delta(x,x')$.
\end{itemize}
We do not consider the cases when $x=y$ or when $x$ is pebbled by the pebble 
removed from $y$ because, in our analysis, we can assume that
Spoiler uses an \emph{optimal} strategy, allowing him to win
the 2-pebble $\Sigma_1$ game on $G_m$ and $H_m$ from the initial position
$(y,y')$ in the smallest possible number of rounds
(if he does not play optimally, Duplicator survives even longer).

\begin{claim}
  \label{cl:discont}
If $x\not\sim y$ and $x\ne y$, then $\Delta(x,x')\le1$.
\end{claim}

\begin{subproof}
Assume first that $x\ne d_0$  and $y\ne d_0$.
W.l.o.g., suppose that $y$ and $y'$ are apricot and, specifically, $y'=a'_j$
(the blue and the cyan cases are symmetric to the apricot case).
Not to lose immediately, Duplicator cannot pebble $x'$ in $\{c'_{j-1},a'_j,b'_j\}$,
where $j-1$ is supposed to be an element of $\integers_{2m}$. 
This can obstruct attaining $\Delta(x,x')=0$
(if $x\in\{c_{j-1},a_j,b_j\}$), but then there is a choice of $x'$ with $\Delta(x,x')=1$.

Assume now that $x=d_0$. Then $x'=d'_0$ if $y'\ne a'_0$ and $x'=d'_1$ otherwise.
In both cases $\Delta(x,x')\le1$.

Finally, let $y=d_0$ and $y'=d'_j$. Then the value $x'=a'_j$ is forbidden and, 
if this prevents $\Delta(x,x')=0$, then we have $\Delta(x,x')=1$.
\end{subproof}

Consider now the dynamical behaviour of $\Delta(x,x')$, assuming that
Duplicator uses the above strategy and Spoiler follows an optimal winning strategy.
We have $\Delta(x,x')=0$ at the beginning of the game and $\Delta(x,x')=m$ at the end
(that is, in the round immediately before Spoiler wins).
Consider the last round of the game where $\Delta(x,x')\le1$.
By Claim \ref{cl:discont}, starting from the next round Spoiler
always moves along an edge in $G_m$. Note that, from now on, visiting
$d_0$ earlier than in the very last round would be inoptimal. Therefore,
Spoiler walks along the circle. Another consequence of optimality is that
he moves always in the same direction. 

W.l.o.g., we can suppose that
Spoiler moves in the ascending order of indices.
Note that $\Delta(x,x')$ increases by 1 only under the transition from
$x=a_{2m-2}$ to $x=a_0$ (at this point, the index of $x$ makes a jump in $\integers_{2m}$,
while the index of $x'$ moves along $\integers_{2m}$ always continuously). 
In order to increase $\Delta(x,x')$ from $1$ to $m$,
the edge $a_{2m-2}a_0$ must be passed $m-1$ times. It follows that,
before Spoiler wins, the game lasts at least $2+(m-2)\cdot3(2m-1)=6m^2-15m+8$ rounds.

Note that $v(G_m)=6m-2$ and $v(H_m)=8m-1$.
In order to make the number of vertices in both graphs $n=8m-1$,
let $m$ be multiple of $3$ and add two new connected components to $G_m$,
namely the cycle of length $2m$ with alternating colors apricot, blue, and cyan
and one isolated vertex of any color. 
Spoiler can still win by playing in the old component.
Since playing in the new components does not help him,
the game on the modified $G_m$ and the same $H_m$ lasts at least
$6m^2-15m+8=\frac3{32}\,n^2-O(n)$ rounds.
\end{proof}

\hide{
\begin{figure}
\centering
 \begin{tikzpicture}
  \FPset{\r}{2}
%  \node at (-\r cm - 1cm,\r cm) {$G_3$};
\begin{scope}[rotate=9,xscale=-1]
  \foreach \i in {0,...,4}
   {
    \node[aavertex] (a\i) at (\i*72:\r cm) {};
    \node[aavertex] (aa\i) at (\i*72+18:\r cm) {};
    \draw (a\i) -- (aa\i);
    \FPsub{\j}{\i}{1}
    \FPround{\j}{\j}{0}
    \ifnumgreater{\i}{0} 
      {\draw (a\i) -- (bb\j);}{}
    \node[bvertex] (b\i) at (\i*72+36:\r cm) {};
    \node[bvertex] (bb\i) at (\i*72+54:\r cm) {};
    \draw (aa\i) -- (b\i) -- (bb\i); 
   }
   \draw (bb4) -- (a0);
   \node[aavertex] (d) at (9:\r + .7) {};
   \draw (a0) -- (d) -- (aa0);
   % \node[right] (egal) at (a0) {$\,\,a_0$};
\end{scope}

  \FPset{\p}{2.2}
%  \node[xshift=70mm] at (-\p cm - .5cm,\r cm) {$H_3$};
\begin{scope}[xshift=70mm,rotate=-9]
  \foreach \i in {0,...,5}
   {
    \node[aavertex] (a\i) at (\i*60:\p cm) {};
    \node[aavertex] (aa\i) at (\i*60+15:\p cm) {};
    \draw (a\i) -- (aa\i);
    \FPsub{\j}{\i}{1}
    \FPround{\j}{\j}{0}
    \ifnumgreater{\i}{0} 
      {\draw (a\i) -- (bb\j);
       \node[aavertex] (d\i) at (\i*60+7.7:\p cm + .7cm) {};
       \draw (a\i) -- (d\i) -- (aa\i);
      }{}
    \node[bvertex] (b\i) at (\i*60+30:\p cm) {};
    \node[bvertex] (bb\i) at (\i*60+45:\p cm) {};
    \draw (aa\i) -- (b\i) -- (bb\i); 
   }
    \draw (bb5) -- (a0);
%   \node[right] (egal) at (a0) {$\,\,a'_0$};
\end{scope}
 \end{tikzpicture}
\caption{A bicolored modification of the example in Fig.~\protect\ref{fig:main-example}.}
\label{fig:bicolor}
\end{figure}
}

\begin{remark}\rm
In order to facilitate the exposition, the construction of graphs $G_m$ and $H_m$
uses 4 colors. In fact, the same idea can be realized with 2 colors.
% Appropriately modified $G_m$ and $H_m$ are shown in Fig.~\ref{fig:bicolor}.
% The larger graph $H_m$ has now $n=10m-1$  vertices, and the game lasts more than
% $\frac2{25}\,n^2-O(n)$ rounds. For simplicity of presentation, we did not try
% to optimize neither the constant factor of $\frac2{25}$ here nor the factor
% of $\frac1{11}$ in Theorem \ref{thm:lower-E}. However, two colors
% are optimal in view of Theorem~\ref{thm:uncolor-exist}.
This is optimal because for uncolored graphs one can show that 
if $G$ is distinguishable from $H$ in
existential two-variable logic, then
$D^2_\E(G,H)\le2\,v(H)$.
\end{remark}

\subsection{Lifting it higher}

Since $D^2_{\Sigma_i}(G,H)=D^2_{\Pi_i}(H,G)$, 
the results of this section 
hold true
as well for $\Pi_i\cap\fo2$.

\begin{theorem}\label{thm:lower-bound}
Let $i\ge1$.
There are infinitely many colored graphs $G$ and $H$, both on $n$ vertices,
such that $G$ is distinguishable from $H$ in $\Sigma_i\cap\fo2$ and
$D^2_{\Sigma_i}(G;H)>\frac1{11\cdot9^i}\,n^2-\frac1{11\cdot3^i}n$.
\end{theorem}

\begin{proof}
For infinitely many values of an integer parameter $n_0$,
Theorem \ref{thm:lower-E} provides us with colored graphs $G_0$ and $H_0$ on $n_0$ vertices each
such that Spoiler has a continuous winning strategy in the 2-pebble $\Sigma_1$ game
on $G_0$ and $H_0$, and $D^2_\E(G_0,H_0)>\frac1{11}\,n_0^2$. Let $G_i$ and $H_i$ be now the graphs
obtained from $G_0$ and $H_0$ by the lifting construction described in Section \ref{sec:prelim}.
Note that $v(G_{i})=3\,v(G_{i-1})+1$, where $G_0=G$. It follows that
$n=v(G_{i})=3^in_0+\frac{3^i-1}2$. The graph
$G_i$ is distinguishable from $H_i$ in $\Sigma_i\cap\fo2$ by part 1 of
Lemma \ref{lem:lifting}. By part 2 of this lemma, we have $D^2_{\Sigma_i}(G_i,H_i)>\frac1{11}\,n_0^2$,
which implies the bound stated in terms of~$n$.
\end{proof}

Note that a similar result can be shown for uncolored directed graphs.
However, Theorem \ref{thm:lower-bound} cannot be extended to uncolored undirected graphs
because in this case one can show that,
if $G$ is distinguishable from $H$ in $\fo2$,
then $D^2_{\Sigma_2}(G,H)\le\max\{v(G),v(H)\}$.

Using a similar sequence of graphs, we can also show that $\Sigma_{i}\cap\fo2$
is more succinct than $\Sigma_{i-1}\cap\fo2$.
Given $i$, let us construct $G_i$ and $H_i$ starting from the same $G_0$ as in the
proof of Theorem \ref{thm:lower-bound} and a slightly modified $H_0$.
Specifically, we make all dandelion vertices in $H_0$ adjacent; see Fig.~\protect\ref{fig:main-example}
for $G_0$ and $H_0$, where $H_0$ is still unmodified.
This makes part 4 of Lemma \ref{lem:lifting} applicable, which along with part 2 gives us
the following result.

\begin{theorem}\label{thm:succinct}
  For each $i\ge2$ there are infinitely many colored graphs $G$ and $H$, 
both on $n$ vertices, such that $D^2_{\Sigma_i}(G,H)=O(1)$ while $D^2_{\Sigma_{i-1}}(G,H)<\infty$
and $D^2_{\Pi_i}(G,H)=\Omega(n^2)$.
\end{theorem}

\section{An upper bound for $D^k_\logic(G,H)$}\label{sec:upper}

Since $D^k_{\Sigma_i}(G,H)=D^k_{\Pi_i}(H,G)$, the following result holds true
as well for $\Pi_i\cap\fo k$.

\begin{theorem}\label{thm:upper-bound}
Let $G$ and $H$ be structures over the same vocabulary.
If $G$ is distinguishable from $H$
in $\Sigma_i\cap\fo k$, then
$D^k_{\Sigma_i}(G,H)\le (v(G)v(H))^{k-1}+1$.
\end{theorem}

\begin{proof}
By Lemma \ref{lem:game}, we have to prove that, if
Spoiler has a winning strategy in the $r$-round $k$-pebble $\Sigma_i$ game on $G$ and $H$
for some $r$, then he has a winning strategy in the game with $v(G)v(H)+1$ rounds.

The proof is based on a general game-theoretic argument.
Consider a two-person game, where the players follow
some fixed strategies and one of them wins.
Then the length of the game cannot exceed the total number
of all possible positions because once a position occurs twice,
the play falls into an endless loop.
Here it is assumed that the players' strategies are \emph{positional},
that is, that a strategy of a player
maps a current \emph{position} (rather than the sequence of all previous positions)
to one of the moves available for the player.

Implementing this scenario for the $\Sigma_i$ game,
we have to overcome two complications. First, we have to ``reduce''
the space $V(G)^k\times V(H)^k$ of all possible positions in the game,
which has size $(v(G)v(H))^k$. Second, we have take care of the fact that,
if $i>1$, then Spoiler's play can hardly be absolutely memoryless in the sense that
he apparently has to remember the number of jumps left to him or, at least,
the graph in which he moved in the preceding round.

We begin with some notation. Let $\baru$ and $\barv$ be tuples of
vertices in $G$ and $H$, respectively, having the same length
no more than $k$. Given $\Xi\in\{\Sigma,\Pi\}$ and $a\ge1$,
let $R(\Xi,a,\baru,\barv)$ be the minimum $r$ such that Spoiler
has a winning strategy in the $\Xi_a$ game on $G$ and $H$ starting
from the initial position $(\baru,\barv)$.
Given a $k$-tuple $\barw$ and $j\le k$, let $\sigma_j\barw$ denote
the $(k-1)$-tuple obtained from $\barw$ by removal of the $j$-th coordinate.
Note that, if $\baru\in V(G)^k$ and $\barv\in V(H)^k$, then
\begin{equation}
  \label{eq:RRs}
 R(\Xi,a,\baru,\barv)=\min_{1\le j\le k}R(\Xi,a,\sigma_j\baru,\sigma_j\barv).
\end{equation}

In order to estimate the length of the $k$-pebble $\Sigma_i$ game on $G$ and $H$,
we fix a strategy for Duplicator arbitrarily and consider the strategy for
Spoiler as described below. For $i\ge1$, we will say that $\barC_s=(\Xi_s,a_s,\baru_s,\barv_s)$ is the
\emph{position after the $s$-th round} if
\begin{itemize}
\item 
$\Xi_s=\Sigma$ if in the $s$-th round Spoiler moved in $G$ and $\Xi_s=\Pi$ if he
moved in $H$;
\item 
during the first $s$ rounds Spoiler jumped from one graph to another $i-a_s$ times;
\item 
after the $s$-th round the pebbles $p_1,\ldots,p_k$ are placed on the vertices
$\baru\in V(G)^k$ and $\barv\in V(H)^k$ (we suppose that in the first round
Spoiler puts all $k$ pebbles on one vertex).
\end{itemize}
Furthermore, we will say that $\tilC_s=(\Xi_s,a_s,\tilu_s,\tilv_s)$ is the
\emph{position before the $(s+1)$-th move} if in the $(s+1)$-th round
Spoiler moves the pebble $p_j$ and $\tilu_s=\sigma_j\baru_s$ and $\tilv_s=\sigma_j\barv_s$.

Let us describe Spoiler's strategy. He makes the first move according to an arbitrarily prescribed
strategy that is winning for him in the $D^k_{\Sigma_i}(G,H)$-round $k$-pebble $\Sigma_i$
game on $G$ and $H$. If this move is in $G$, let $\Xi_1=\Sigma$ and $a_1=i$;
otherwise $\Xi_1=\Pi$ and $a_1=i-1$. After Duplicator responses, the position $\barC_1$ is specified.
Note that $R(\barC_1)<D^k_{\Sigma_i}(G,H)$.

Suppose that the $s$-th round has been played and after this we have the position $\barC_s=(\Xi_s,a_s,\baru_s,\barv_s)$.
In the next round Spoiler plays with the pebble $p_j$ for the smallest value of $j$ such that
\begin{equation}
  \label{eq:tilRbar}
R(\tilC_s)=R(\barC_s).  
\end{equation}
Such index $j$ exists by \refeq{RRs}. Spoiler makes his move according to a prescribed
strategy that is winning for him in the $R(\barC_s)$-round $k$-pebble $(\Xi_s)_{a_s}$
game on $G$ and $H$ with the initial position $(\tilu_s,\tilv_s)$.
If he moves in the same graph as in the $s$-th round, then 
$\Xi_{s+1}=\Xi_s$ and $a_{s+1}=a_s$; otherwise $\Xi_{s+1}$ gets flipped and $a_{s+1}=a_s-1$.

Note that $a_{s+1}\le a_s$ and, if $\Xi_{s+1}\ne\Xi_s$, then $a_{s+1}<a_s$.
Since Spoiler in each round uses a strategy optimal for the rest of the game,
\begin{equation}
  \label{eq:RsRss}
 R(\barC_{s+1})<R(\barC_s).
\end{equation}
It follows that the described strategy allows Spoiler to win the $\Sigma_i$
game on $G$ and $H$ in at most $D^k_{\Sigma_i}(G,H)$ moves.

We now estimate the length of the game from above.
Suppose that after the $t$-th round Duplicator is still alive.
Due to \refeq{tilRbar} and \refeq{RsRss}, 
$$
R(\tilC_1)>R(\tilC_2)>\ldots>R(\tilC_t).
$$
It follows that the elements of the sequence $\tilC_1,\tilC_2,\ldots,\tilC_t$
are pairwise distinct. We conclude from here that the elements of the sequence 
$(\tilu_1,\tilv_1),(\tilu_2,\tilv_2),\ldots,(\tilu_t,\tilv_t)$
are pairwise distinct too. Indeed, let $s'>s$. If $a_{s'}=a_s$, then
$\Xi_s=\Xi_{s'}$. Since $\tilC_s\ne\tilC_{s'}$, we have $(\tilu_s,\tilv_s)\ne(\tilu_{s'},\tilv_{s'})$. 
If $a_{s'}<a_s$, the same inequality follows from the fact that
$R(\Xi,a,\tilu,\tilv)\le R(\Xi',a',\tilu,\tilv)$ whenever $a'<a$.

Since $(\tilu_s,\tilv_s)$ ranges over $V(G)^{k-1}\times V(H)^{k-1}$, we conclude that
$t\le(v(G)v(H))^{k-1}$ and, therefore, Spoiler wins in the round $(v(G)v(H))^{k-1}+1$
at latest.
\end{proof}

\section{The collapse of the alternation hierarchy of $\fo2$ over uncolored graphs}\label{sec:uncol-collapse}

We here show that the quantifier alternation hierarchy of $\fo2$ over uncolored graphs
collapses to the second level.

\begin{theorem}\label{thm:uncol-collapse}
  If a class of uncolored graphs is definable by a first-order formula with two variables,
then it is definable by a first-order formula with two variables and one quantifier alternation. 
\end{theorem}

The proof takes the rest of this section.

The {\em complement\/} of a graph $G$
is the graph on the same vertex set $V(G)$ with any two vertices adjacent if
and only if they are not adjacent in~$G$.
We call a graph \emph{normal} if it has neither isolated nor universal vertex.
Note that a graph is normal iff its complement is normal.
For every graph $G$ with at least 2 vertices we inductively define its \emph{rank}~$\rank G$.
\begin{itemize}
\item 
Graphs of rank 1 are exactly the empty, the complete, and the normal graphs.
\item 
Graphs of rank 2 are exactly the graphs obtained by adding universal vertices
to empty graphs, or isolated vertices to complete graphs, or either universal or
isolated vertices to normal graphs.
\item 
If $i\ge2$,
disconnected graphs of rank $i+1$ are obtained from connected graphs of rank $i$ 
by adding a number of isolated vertices.
\item 
For every $i$,
connected graphs of rank $i$ are exactly complements of
disconnected graphs of rank $i$.
\end{itemize}

A simple inductive argument on the number of vertices shows that all graphs
with at least two vertices get ranked. Indeed, if a graph $G$ is normal,
complete, or empty, it receives rank 1. This includes the case that $G$
has two vertices. If $G$ does not belong to any of these three classes,
it has either isolated or universal vertices. 
Since graphs with universal vertices are connected and are the complements
of graphs with isolated vertices, it suffices to consider the case
that $G$ has isolated vertices. Remove all of them from $G$
and denote the result by $G'$. Note that $G'$ has less vertices than $G$
but still more than one vertex. By the induction assumption, $G'$ is ranked.
If $\rank G'=1$, then $\rank G=2$ by definition.
If $\rank G'>1$ (i.e., $G'$ is complete or normal), then $G'$ must be connected (for else it would be normal).
Therefore, $\rank G=\rank G'+1$ by definition. 

We now introduce a ranking of vertices in a graph $G$.
If $\rank G=1$, then all vertices of $G$ get \emph{rank $1$}.
Suppose that $\rank G>1$.
If $G$ is disconnected, it has at least one isolated vertex;
if $G$ is connected, there is at least one universal vertex.
Denote the set of such vertices by $\partial G$.
Every vertex in $\partial G$ is assigned \emph{rank $1$}.
If $u\notin\partial G$, then it is assigned \emph{rank} one greater than
the rank of $u$ in the graph $G-\partial G$.
The rank of a vertex $u$ in $G$ will be denoted by $\rank u$.
It ranges from the lowest value $1$ to the highest value~$\rank G$.
Note that a vertex $u$ with $\rank u<\rank G$ has the same adjacency to all other vertices 
of equal or higher rank.

Given an integer $m\ge1$ and a graph $G$ with $\rank G>m$, we define
the \emph{$m$-tail type} of $G$ to be the sequence $(t_0,t_1,\ldots,t_m)$ where
$t_0\in\{\conn,\disc\}$ depending on whether $G$ is connected or disconnected
and, for $i\ge1$, $t_i\in\{\thin,\thick\}$ depending on whether $G$ has one or more
vertices of rank~$i$.

Furthermore, we define
the \emph{kernel} of a graph $G$ to be its subgraph induced on the vertices of rank $\rank G$.
Note that the kernel of any $G$ is a graph of rank 1. We define the \emph{head type} of $G$ to be
$\empt$, $\compl$, or $\norma$ depending on the kernel.
We say that graphs $G$ and $H$ are of the same \emph{type} if 
$\rank G=\rank H$, $G$ and $H$ have the same head type, and if $\rank G>1$,
then they also have the same $m$-tail type for $m=\rank G-1$.
The single-vertex graph has its own type.

\begin{lemma}\label{lem:indist}\mbox{}

  \begin{enumerate}
\item 
If $G$ and $H$ are of the same type, then $D^2(G,H)=\infty$.
  \item 
  If $G$ and $H$ have the same $m$-tail type, then $D^2(G,H)\ge m$.
  \end{enumerate}
\end{lemma}

\begin{lemma}\label{lem:dist}\mbox{}

  \begin{enumerate}
  \item 
For each $m$-tail type, the class of graphs of this type is definable by
a first-order formula with two variables and one quantifier alternation.
\item 
For each $G$, the class of graphs of the same type as $G$ is definable by
a first-order formula with two variables and one quantifier alternation. 
  \end{enumerate}
\end{lemma}

Let $C$ be a class of graphs definable by a formula with two variables
of quantifier depth less than $m$. By Lemma \ref{lem:indist},
$C$ is the union of finitely many classes of graphs of the same type
(each of rank at most $m$)
and finitely many classes of graphs of the same $m$-tail type.
By Lemma \ref{lem:dist},
$C$ is therefore definable by a first-order formula with 
two variables and one quantifier alternation.
To complete the proof of Theorem \ref{thm:uncol-collapse},
it remains to prove the lemmas.

\begin{proofof}{Lemma \ref{lem:indist}}
{\bf 1.}~Let $\rank G=\rank H=m+1$. Let $V(G)=U_1\cup\ldots U_{m+1}$
and $V(H)=V_1\cup\ldots V_{m+1}$ be the partitions of the vertex
sets of $G$ and $H$ according to the ranking of vertices.
We will describe a winning strategy for Duplicator in the two-pebble
game on $G$ and $H$. We will call a pair of pebbled vertices
$(u,v)\in V(G)\times V(H)$ \emph{straight} if $u\in U_i$ and $v\in V_i$
for the same $i$. Note that both the kernels $U_{m+1}$ and $V_{m+1}$
contain at least 2 vertices and, since $G$ and $H$ are of the same type,
$|U_i|=1$ iff $|V_i|=1$. This allows Duplicator to play so that
the vertices pebbled in each round form a straight pair and the equality
relation is never violated. If the head type of $G$ and $H$ is 
$\empt$ or $\compl$, this strategy is winning because the adjacency of
vertices $u\in U_i$ and $u'\in U_j$ depends only on the indices $i$ and $j$
and is the same as the adjacency of any vertices $v\in V_i$ and $v'\in V_j$.
It remains to notice that Duplicator can resist also when the game
is played inside the normal kernels $U_{m+1}$ and $V_{m+1}$.
In this case she never loses because, for every vertex in a normal graph, she can
find another adjacent or non-adjacent vertex, as she desires.

{\bf 2.}
We have to show that Duplicator can survive in at least $m-1$ rounds.
Note that both $\rank G\ge m+1$ and $\rank H\ge m+1$.
Similarly to part 1, consider partitions $V(G)=U_1\cup\ldots U_{m+1}$
and $V(H)=V_1\cup\ldots V_{m+1}$, where $U_{m+1}$ and $V_{m+1}$ now consist
of the vertices whose rank is higher than $m$.
In the first round Duplicator plays so that the pebbled vertices form a straight pair.
However, starting from the second round it can be for her no more possible
to keep the pebbled pairs straight.
Call a pair of pebbled vertices
$(u,v)\in V(G)\times V(H)$ \emph{skew} if $u\in U_i$ and $v\in V_j$
for different $i$ and $j$. 
Assume that Spoiler uses his two pebbles alternatingly
(playing with the same pebble in two successive rounds gives him no
advantage). Let $(u_r,v_r)$ denote the
pair of vertices pebbled the $r$-th round. If $(u_r,v_r)$ is skew,
let $S_r$ denote the minimum $s$ such that $u_r\in U_s$ or $v_r\in V_s$. 
If $(u_r,v_r)$ is straight, we set $S_r=m+1$.
Our goal is to show that, if $S_r=m+1$, then Duplicator has a non-losing move in the next round
such that $S_{r+1}\ge m-1$ and that, as long as $1<S_r\le m$, she has a non-losing move
such that $S_{r+1}\ge S_r-1$. This readily implies that Duplicator does not lose
the first $m-1$ rounds.

To avoid multiple treatment of symmetric cases, we use the following notation.
Let $\{G_1,G_2\}=\{G,H\}$. Let $y_1\in G_1$ and $y_2\in G_2$ denote the vertices
being pebbled in the round $r+1$, and let $x_1\in G_1$ and $x_2\in G_2$ be the vertices
pebbled  in the round $r$ (in the previous notation, $\{x_1,x_2\}=\{u_r,v_r\}$ and 
$\{y_1,y_2\}=\{u_{r+1},v_{r+1}\}$). 

Suppose first that $\{x_1,x_2\}$ is a straight pair contained in the slice $U_i\cup V_i$.
If $i\le m$, it makes no problem for Duplicator to move so that the pair $\{y_1,y_2\}$
is also straight. This holds true also if $i=m+1$ and Spoiler pebbles $y_a\in U_j\cup V_j$
with $j\le m$. Thus, in these cases $S_{r+1}=S_r=m+1$. However, if $i=j=m+1$,
moving straight can be always Duplicator's loss. In this case she survives by
pebbling a vertex $y_{3-a}$ of rank $m$ or $m-1$, depending on the adjacency
relation between $x_a$ and $y_a$. In this case $S_{r+1}\ge m-1$.

Let us accentuate the property of the vertex ranking that is beneficial to Duplicator
in the last case. Recall that, if a vertex $u$ is not in the graph kernel,
it has the same adjacency to all other vertices 
of equal or higher rank. If $u$ is adjacent to all such vertices,
we say that $u$ is of \emph{universal type}; otherwise we say that 
it is of \emph{isolated type}. Duplicator uses the fact that the type
of a vertex gets flipped when its rank increases by one.

Suppose now that $\{x_1,x_2\}$ is a skew pair. 
Let $x_1\in U_i\cup V_i$ and $x_2\in U_j\cup V_j$ and,
w.l.o.g., assume that $i>j$.
Since $j=S_r$, it is supposed that $j>1$.
We consider three cases depending on Spoiler's move $y_a$.
In the most favorable for Duplicator case, $\rank y_a<j$.
Then Duplicator responds with a vertex $y_{3-a}$ of the same rank,
resetting $S_{r+1}$ back to the initial value $m+1$.
If Spoiler pebbles a vertex $y_2$ of $\rank y_2\ge j$, then Duplicator responds with 
a vertex $y_1$ of $\rank y_1=j$, keeping $S_{r+1}\ge j=S_r$ (unchanged or reset to $m+1$).
Finally, consider the case when Spoiler pebbles a vertex $y_1$ of $\rank y_1\ge j$.
Assume that $x_2$ is of universal type (the other case is symmetric).
If $y_1$ and $x_1$ are adjacent, then Duplicator responds with 
a vertex $y_2$ of $\rank y_2=i$, keeping $S_{r+1}\ge S_r$.
If $y_1$ and $x_1$ are not adjacent, then Duplicator responds with 
$y_2$ of $\rank y_2=j-1$, which is of isolated type.
This is the only case when $S_{r+1}=S_r-1$ decreases.
\end{proofof}

\begin{proofof}{Lemma \ref{lem:dist}}
{\bf 1.}~Consider an $m$-tail type $(t_0,t_1,\ldots,t_m)$.
Assume that $t_0=\conn$ (the case of $t_0=\disc$ is similar).
Let $\sim$ denote the adjacency relation.
We inductively define a sequence of formulas $\Phi_s(x)$
with occurrences of two variables $x$ and $y$ and with one free variable:
\begin{eqnarray*}
  \Phi_1(x)&\feq&\A y\,(y\sim x\vee y=x),\\
 \Phi_{2k}(x)&\feq&\A y\,(\Phi_{2k-1}(y)\vee{y\not\sim x}),\\
\Phi_{2k+1}(x)&\feq&\A y\,(\Phi_{2k}(y)\vee{y\sim x}\vee{y=x}).
\end{eqnarray*}
Here $\Phi_{2k-1}(y)$ is obtained from $\Phi_{2k-1}(x)$ by swapping $x$ and $y$.
A simple inductive argument shows that, if $G$ is a connected graph and $\rank G$
is greater than an odd (resp.\ even) integer $s$, then $G,v\models\Phi_s(x)$ exactly when
the vertex $v$ is of universal (resp.\ isolated) type and $\rank v\le s$.

Furthermore, we define a sequence of closed formulas $\Psi_s$ 
with alternation number~1:
\begin{eqnarray*}
  \Psi_1&\feq&\E x\Phi_1(x)\und\E x\neg\Phi_1(x),\\
 \Psi_2&\feq&\E x\Phi_1(x)\und\E x\Phi_2(x)\und\E x(\neg\Phi_1(x)\und\neg\Phi_2(x)),\\
\Psi_s&\feq&\E x\Phi_1(x)\und\E x\Phi_2(x)\und\bigwedge_{i=3}^s\E x(\Phi_i(x)\und\neg\Phi_{i-2}(x))
%\\&&
\und\E x(\neg\Phi_{s-1}(x)\und\neg\Phi_s(x)),
\end{eqnarray*}
 where $s\ge3$.
Note that a graph $G$ satisfies $\Psi_s$
if and only if $G$ is connected and $\rank G>s$. 

We are now able to define the class of graphs of $m$-tail type $(t_0,t_1,\ldots,t_m)$
by the conjunction
$$
\Psi_m\und\bigwedge_{i=1}^m\Tau_i,
$$
where
$$
\Tau_i\feq\E x\E y\,({x\ne y}\und \Phi_i(x)\und\neg\Phi_{i-2}(x)\und\Phi_i(y)\und\neg\Phi_{i-2}(y))
$$
if $t_i=\thick$ and
$$
\Tau_i\feq\A x\A y\,(\neg\Phi_i(x)\Or\neg\Phi_i(y)\Or\Phi_{i-2}(x)\Or\Phi_{i-2}(y)\Or{x=y})
$$
if $t_i=\thin$ (if $i\le2$, the subformulas with non-positive indices should be ignored).

{\bf 2.}
The single-vertex graph is defined by a formula $\A x\A y\,(x=y)$.
The three classes of graphs of rank 1 are defined by the following three
formulas:
\begin{eqnarray*}
&&\E x\E y\,(x\ne y)\und\A x\A y\,(x\not\sim y),\\
&&\E x\E y\,(x\ne y)\und\A x\A y\,({x=y}\Or{x\sim y}),\\
&&\A x\E y\,(x\sim y)\und\A x\E y\,({x\ne y}\Or{x\not\sim y}).
\end{eqnarray*}
Suppose that $\rank G=m+1$ and $m\ge1$. Let $(t_0,t_1,\ldots,t_m)$
be the $m$-tail type of $G$.
Assume that $G$ is connected, that is, $t_0=\conn$ (the disconnected case is similar).
We use the formulas $\Phi_s(x)$, $\Psi_s$, and $\Tau_i$ constructed in the first part.
If the head type of $G$ is $\compl$ or $\empt$ (the former is possible if $m$
is even and the latter if $m$ is odd), then the type of $G$ is defined by
$$
\Psi_m\und\bigwedge_{i=1}^m\Tau_i\und\A x\Phi_{m+1}(x).
$$
If the head type of $G$ is $\norma$, then the type of $G$ is defined by
$$
\Psi_m\und\bigwedge_{i=1}^m\Tau_i\und\neg\E x\Phi_{m+1}(x).
$$
Indeed, $\Psi_m\und\bigwedge_{i=1}^m\Tau_i$ is true on a graph $H$ if and only if 
$H$ has the $m$-tail type $(t_0,t_1,\ldots,t_m)$ and $\rank H\ge m+1$.
Let $Q\subset V(H)$ denote the set of vertices not in the tail part.
Then $Q$ is a homogeneous set exactly when
$H$ satisfies $\A x\Phi_{m+1}(x)$, and $Q$ spans a normal subgraph exactly when
$H$ satisfies $\neg\E x\Phi_{m+1}(x)$.
\end{proofof}

\begin{remark}\rm
Part 1 of Lemma \ref{lem:indist} and part 2 of Lemma \ref{lem:indist}
readily imply that two uncolored graphs are indistinguishable in $\fo2$
if and only if they have the same type. The type of a given graph $G$
can be determined in \tc0. This follows from two observations:
\begin{itemize}
\item 
Suppose that $v$ is in the kernel of $G$ while $u$ and $w$ are not, i.e.,
$\rank v=\rank G$ and $\rank u,\rank w<\rank G$. If $u$ is of universal type
and $w$ is of isolated type, then $\deg w<\deg v<\deg u$.
\item 
The degree of a vertex of isolated type increases 
together with its rank, while the degree of a vertex of universal type decreases
as its rank increases.
\end{itemize}
Based on these observations, the type of $G$ is easy (in $\ac0$) to determine
after computing all vertex degrees (in $\tc0$).
It follows that the equivalence problem for $\fo2$ over uncolored graphs is in $\tc0\subseteq\nc1$.
If extended to $\fo3$ or to colored graphs, the equivalence problem is known to be
P-complete (Grohe \cite{Grohe99}).
\end{remark}

\begin{remark}\rm
The $\fo2$-equivalence relation over trees degenerates to four classes whose
represenatives are, respectively, the path graphs $P_1$, $P_2$, $P_3$, and $P_4$.
Here $P_n$ denotes the path on $n$ vertices. Note that $P_1$ is the single-vertex graph,
$P_2$ is complete, $P_4$ is normal, and $P_3$ has rank 2. The equivalence classes
of $P_1$ and $P_2$ are singletons. The equivalence class of $P_3$ consists of 
the stars with at least 3 vertices. The class of $P_4$ consists of all remaining trees,
which are normal. As it is easy to see, any two inequivalent trees are distinguishable
by a $\fo2$ formula with no quantifier alternation.
\end{remark}

\begin{remark}\rm
In general, Theorem \ref{thm:uncol-collapse} is optimal with respect to the alternation number.
Let $n\ge 5$ and consider the cycle $C_n$ on $n$ vertices and the wheel graph $W_n$ consisting of
the cycle $C_{n-1}$ and a universal vertex. Since $\rank C_n=1$ and $\rank W_n=2$,
these graphs are distinguishable in $\fo2$. However, if Spoiler plays only in $C_n$,
Duplicator wins by moving in the $C_{n-1}$ subgraph of $W_n$,
which is a normal graph. If Spoiler plays only in $W_n$,
Duplicator wins because $C_n$ is normal. This example justifies the equality in 
the 6-th row of Figure~\ref{fig:afunc}.
\end{remark}

\section{Discussion and further questions}

We left several questions open.
Is the logarithmic lower bound for $A^2(n)$ of Theorem \ref{thm:hierarchy2-color} tight over trees?
Can the general upper bound $A^2(n)\le n+1$ be improved? Say, to $A^2(n)\le n/2$?

Some of our arguments for $\fo2$ seem not to have obvious extensions to the larger
number of variables. 
Can the  logarithmic lower bound for $A^3(n)$ be improved over colored graphs? 
Say, to $A^3(n)=\Omega(n)$?
How tight is the upper bound $D^k_{\E}(G,H)\le (v(G)v(H))^{k-1}+1$ 
(cf.\ Theorem \ref{thm:upper-bound}) if $k\ge3$?

Is there a trade-off between the number of variables and the number of alternations?
Specifically, suppose that $n$-element structures $G$ and $H$ are distinguishable in $\fo2$.
In other terms, we know that $A^2(G,H)$ is finite, though it can be linear by Theorem \ref{thm:altcount}.
Does there exist a constant $k$ such that $A^k(G,H)=O(\log n)$ for any such $G$ and $H$?
Similarly, can we reduce even the quantifier depth $D^k(G,H)$ if we take a larger $k$?
A negative result in this direction is obtained
by Immerman \cite{Immerman81} who constructed pairs of colored digraphs $G$ and $H$
such that $D^3(G,H)<\infty$ while $D_{\fo{}}(G,H)>2^{\sqrt{\log n}-1}$.

Is there a trade-off between the number of alternations and the quantifier depth?
Specifically, let $D^k_a(G,H)=D^k_{\Sigma_{a+1}\cup\Pi_{a+1}}(G,H)$.
Suppose that $D^2_\E(G,H)<\infty$. 
By Theorem \ref{thm:lower-E} we know that this can require quadratic quantifier depth,
that is, sometimes we can have $D^2_\E(G,H)=\Omega(n^2)$. Is
there a constant $a$ such that $D^2_a(G,H)=O(n)$ for all such $G$ and $H$?
Moreover, can we exclude that even $a=1$ suffices? Or the lower bound
of Theorem \ref{thm:lower-E} can be strengthened to $D^2_0(G,H)=\Omega(n^2)$?

Let $D^k_a(G)$ denote the minimum quantifier depth of a formula in $(\Sigma_{a+1}\cup\Pi_{a+1})\cap\fo k$
defining a graph $G$ up to isomorphism.
Note that $D^k_a(G)=\max_HD^k_a(G,H)$, where the maximum
is taken over all graphs non-isomorphic to $G$.
Theorem \ref{thm:succinct} shows examples of a large gap between
the parameters $D^k_a(G,H)$ and $D^k_{a+1}(G,H)$
How far apart
can the parameters $D^k_a(G)$ and $D^k_{a+1}(G)$ be from one another?
If we consider the similar parameters for the unbounded-variable logic $\fo{}$, 
then it is known \cite{PikhurkoSV06} that the gap between $D_3(G)$ and $D_0(G)$
can be huge: there is no total recursive function such that $D_0(G)\le f(D_3(G))$.

\end{document}